\documentclass[11pt,a4paper]{article}

\usepackage{braket} 
\usepackage{mathtools}
\usepackage{authblk} 
\usepackage{algpseudocode,algorithmicx,algorithm}

\usepackage{mathrsfs}
\usepackage{latexsym,bm}
\usepackage{amsmath,amsfonts,amsmath,amssymb,amsthm}
\usepackage{breqn}
\usepackage{extarrows}
\usepackage{float}
\usepackage[T1]{fontenc} 
\usepackage{ragged2e}

\usepackage{listings}
\usepackage{xcolor}

\usepackage{mathrsfs}
\usepackage{qcircuit}
\usepackage{graphicx,epstopdf,float}
\usepackage{enumerate,cases,multirow}
\usepackage{makecell}
\usepackage[caption=false]{subfig}
\usepackage{subfig}
\usepackage{pgfplots}

\usepackage{longtable,colortbl,arydshln,threeparttable}
\definecolor{mygray}{gray}{.9}
\usepackage[shortlabels]{enumitem}

\usepackage{indentfirst}
\usepackage[top=25mm,bottom=20mm,left=25mm,right=20mm]{geometry}
\usepackage{cite}

\usepackage{listings}

\usepackage{makeidx}        
\usepackage{booktabs}
\usepackage[bookmarksnumbered,colorlinks,citecolor=red,linkcolor=red,hyperindex,linktocpage=true]{hyperref}
\usepackage{subfiles}

\newtheorem{theorem}{Theorem}[section]
\newtheorem{lemma}{Lemma}[section]

\theoremstyle{remark}
\newtheorem{remark}{\bf Remark}[section]

\numberwithin{equation}{section}

\begin{document}
\title{Quantum circuits of elastic wave equations via Schr$\ddot{\text{o}}$dingerisation}

\author[1,2,3]{Shi Jin   \thanks{shijin-m@sjtu.edu.cn}}
\author[1]{Chundan Zhang   \footnote{Corresponding author.} \thanks{chundanzhang@sjtu.edu.cn}}
\affil[1]{School of Mathematical Sciences,   Shanghai Jiao Tong University, Shanghai 200240, China.}
\affil[2]{Institute of Natural Sciences, Shanghai Jiao Tong University, Shanghai 200240, China.}
\affil[3]{Ministry of Education Key Laboratory in Scientific and Engineering Computing, Shanghai Jiao Tong University, Shanghai 200240, China.}

\date{\today}
\maketitle
\begin{abstract}
We construct explicit quantum circuits for the velocity--stress formulation
of elastic wave equations in homogeneous media using
Schr\"odingerisation. After rewriting the equations in symmetric matrix form(SMF),
we use Fourier spectral discretization for source-driven periodic problems and staggered-grid
finite differences for Dirichlet, Neumann, and impedance boundary conditions.
The Schr\"odingerisation method transforms
any linear PDEs into a system of Schr\"odinger-type PDEs –with unitary evolution–
using the warped phase transformation that maps the equations in one higher dimension.
For periodic problems, we develop circuits by combining controlled Fourier-mode rotations.
For physical boundaries, interleaved encoding swaps the spatial and field-component registers, enabling a more efficient circuit construction using shift operators. 
We analyze $p$th-order Trotter--Suzuki formulas and derive splitting-error and
CNOT-gate estimates, under the assumptions of smooth auxiliary initialization and solution-state recovery. Comparison with classical partitioned splitting
methods identifies a polynomial improvement in spatial-resolution dependence
for the quantum state-evolution subroutine.
Noiseless circuit simulations of a forced periodic problem and an
impedance problem illustrate the constructions and quantify their errors
against exact solutions and classical reference computations.
\end{abstract}


\section{Introduction}

Elastic wave equations describe the propagation of mechanical disturbances in solids~\cite{navier1827} and are widely used in seismology and
geodynamics~\cite{bullen1985introduction}. The velocity--stress formulation gives a first-order hyperbolic system that is well suited to staggered-grid discretization and the treatment of physical boundary conditions~\cite{dai1995wave}. In practical problems, Dirichlet, Neumann, and impedance conditions model fixed, traction, and absorbing boundaries, while source terms represent external forces or seismic excitations~\cite{kaser2006arbitrary}. Therefore, it is important to develop efficient algorithms for elastic systems with source terms and physically relevant boundary conditions.

Classical numerical methods, such as finite difference methods \cite{leveque2007finite}, finite element methods \cite{finite-element}, boundary element methods \cite{boundaryelement} and other numerical methods\cite{FGAChai2018}
, have been extensively developed for elastic wave
simulation. Nevertheless, large-scale simulations in high dimensions, long-time
wave propagation, and models with fine spatial resolution or complex media may
still lead to substantial computational cost. 
This has motivated the development of quantum algorithms for partial differential equations. Early approaches mainly relied on quantum linear system algorithms, such as HHL-type methods, by discretizing the PDE and reformulating it as a linear system \cite{Ber14,Childs20}. More recently, Hamiltonian-simulation-based approaches have provided a more natural and advantageous framework for PDEs. The Schr\"odingerisationembeds general non-unitary dynamics into a unitary evolution in an enlarged Hilbert space \cite{Schrodingerization,JLY22SchrShort}, which is the only method that transfer general linear PDEs into unitary PDEs, thus also most natural for continuous quantum computing \cite{analogPDE}.
A complementary approach is the  linear combination of Hamiltonian simulation (LCHS) method for simulating non-unitary dynamics \cite{an2023linear}. Based on these frameworks, explicit quantum circuits have been developed for several classes of PDEs, including hyperbolic equations \cite{sato2024hamiltonian} and canonical PDEs in Fourier space \cite{Fourierquantum}. These techniques have also been applied to elastic wave propagation, for which explicit quantum algorithms and circuit constructions have been investigated \cite{Sato3Delastic}. These developments demonstrate the importance of exploiting the algebraic and differential structure of PDEs to obtain implementable quantum circuits with explicit complexity estimates.

A useful framework for connecting general linear PDEs with Hamiltonian simulation
is Schr\"odingerisation method\cite{Schrodingerization, JLY22SchrShort}. The main idea is to transform a non-Schr\"odinger or
non-unitary evolution into an enlarged Schr\"odinger-type system with unitary evolution, so that
Hamiltonian simulation techniques can be applied. The formulation is natural for PDEs and  other realistic physical models. Explicit Schr\"odingerisation circuits have been developed for
advection and heat equations~\cite{hu2024quantum}, Maxwell's
equations~\cite{ma2024schr}, Black--Scholes equations~\cite{bs}, and heat
equations with physical boundary conditions~\cite{jin2025quantum}. Its
application to elastic wave equations has also been
studied~\cite{elastic}. For the velocity--stress system, the symmetric matrix
form (SMF)~\cite{SMFd3,SMFd2} exposes the coupling between physical components
and spatial derivatives. The source-free periodic system admits unitary
evolution after discretization, whereas forcing and physical boundary terms
require a treatment of the resulting non-unitary dynamics.

In this paper, we develop explicit circuits for these semidiscrete systems. We use Fourier spectral discretization under periodic
boundary conditions and staggered-grid finite differences under physical boundary conditions. 
The circuits are assembled from elementary evolutions using nested product formulas. 
We further analyze the error contributions arising from the Hamiltonian simulation and the spatial and auxiliary-variable discretizations. 

Smooth auxiliary initialization~\cite{matrixqueries} is used
in the complexity analysis to control the auxiliary Fourier resolution.
The main contributions are as follows:
\begin{enumerate}
\item[(a)] We construct quantum circuits for periodic velocity--stress systems with
time-independent sources. Fourier
diagonalization gives controlled rotations for the spatial and auxiliary
variables. For \(d=2,3\), we also provide explicit circuit constructions and estimates of the corresponding CNOT-gate complexity.
\item[(b)] We derive circuits for physical boundary conditions using
exchanged encoding for the spatial and the field-component registers. The Neumann
circuit follows by conjugating the Dirichlet spatial circuit, while the
impedance construction uses exact controlled rotations on the two boundary
subspaces.
\item[(c)] We derive commutator-based error bounds and CNOT-gate estimates for $p$th-order Trotter--Suzuki formulas, including the effects of auxiliary
approximation. Comparison with classical methods identifies a polynomial advantage in spatial
resolution for the state-evolution subroutine under the stated assumptions.
\item[(d)]The circuits experiments are performed using the {\it Unitarylab} package \cite{unitarylab2025}, which facilitates the construction and execution of the quantum circuits on a local simulator. Each test compares the recovered
circuit solution with the exact PDE solution, a classical partitioned
solution, and a classical Schr\"odingerised solution.
\end{enumerate}

Our setting differs from Refs.~\cite{Fourierquantum,Sato3Delastic} mainly in
its modelling scope. Ref.~\cite{Fourierquantum} treats the isotropic acoustic wave equation, whereas we consider the coupled elastic velocity--stress system with both compressional and shear waves. 
Ref.~\cite{Sato3Delastic} considers a source-free elastic system with a specific zero boundary treatment, while our framework, enabled by Schr\"odingerisation to embed the non-unitary dynamics into a unitary evolution, includes periodic conditions with sources and Dirichlet, Neumann, and impedance boundary conditions.

The paper is organized as follows. Section~\ref{sec:Equ review} reviews the velocity--stress
equations and their symmetric matrix form. Section~\ref{sec:schro_review} introduces
Schr\"odingerisation, recovery, and smooth auxiliary initialization.
Section~\ref{sec:Periodic} develops the circuits of the source-driven system and their complexity estimates
in one and multiple dimensions. Section~\ref{sec:Physical} treats physical
boundary conditions. Section~\ref{sec:Experiments} presents the numerical experiments, and
Section~\ref{sec:Conclusion} concludes the paper.

Throughout this paper, \(\|\cdot\|\) denotes the spectral norm. \(N\) denotes the number of discretization points, whereas \(n\) denotes the corresponding number of qubits. Their specific meanings are indicated by subscripts, e.g., \(N_x\) and \(n_p\). For notational simplicity, the tensor power in \(I^{\otimes n}\) is omitted whenever its dimension can be inferred from the context.

\section{A review of velocity--stress equation}\label{sec:Equ review}
The motion of a small displacement  in  an elastic medium can be described as:
\begin{subequations}\label{eq:velocity_stress}
\begin{align}
&\rho \partial_t v_i
    =
    \sum_{j=1}^d \partial_{x_j}\sigma_{ij}
    +
    f_i,
    \qquad i=1,\ldots,d,\label{eq:v}
\\
&\partial_t \sigma_{ii}
    =
    (\lambda+2\mu)\partial_{x_i}v_i
    +
    \lambda\sum_{j\neq i}\partial_{x_j}v_j,
    \qquad i=1,\ldots,d,\label{eq:sigma_nn}
\\
&\partial_t \sigma_{ij}
    =
    \mu
    \left(
        \partial_{x_j}v_i+\partial_{x_i}v_j
    \right)
    +
    g_{ij},
    \qquad 1\leq i<j\leq d.\label{eq:sigma_ij}
\end{align}
\end{subequations}
where $v_i$ represents the velocity in the $i$-th direction, $\sigma_{ij}$ denotes the components of the Cauchy stress tensor, representing the force acting in the $i$-th direction on a surface with normal in the $j$-th direction. $ f_i$ is the body force. $\lambda$ and $\mu$ describe the properties of the medium called Lam$\acute{e}$ parameters.

Let \(\sigma_n=(\sigma_{11},\ldots,\sigma_{dd})^T\) be the vector of normal stresses. The normal-stress part is governed by the
matrix
\[
    C_d=2\mu I_d+\lambda \mathbf 1\mathbf 1^T,
    \qquad
    \mathbf 1=(1,\ldots,1)^T.
\]
Indeed, the normal-stress equations \eqref{eq:sigma_nn} can be written as
\begin{equation}\label{eq:sigma_n matrix}
    \partial_t \sigma_n
    =
    C_d
    \begin{pmatrix}
    \partial_{x_1}v_1\\
    \vdots\\
    \partial_{x_d}v_d
    \end{pmatrix},
\end{equation}
Introducing the rescaled variables
\begin{equation}\label{eq:rescaledvariables}
    \zeta=C_d^{-1/2}\sigma_n,
    \qquad
    s_{ij}=\frac{\sigma_{ij}}{\sqrt{\mu}},
    \qquad
    u_i=\sqrt{\rho}\,v_i.    
\end{equation}
For simplicity, we still denote the rescaled velocity \(u_i\) by \(v_i\) in the
following.  Since
\[
    C_d
    =
    2\mu
    \left(
        I_d-\frac{1}{d}\mathbf 1\mathbf 1^T
    \right)
    +
    (d\lambda+2\mu)
    \frac{1}{d}\mathbf 1\mathbf 1^T,
\]
one gets 
\[
    C_d^{1/2}
    =
    \sqrt{2\mu}
    \left(
        I_d-\frac{1}{d}\mathbf 1\mathbf 1^T
    \right)
    +
    \sqrt{d\lambda+2\mu}
    \frac{1}{d}\mathbf 1\mathbf 1^T.
\]
Thus the entries of \(\rho^{-1/2}C_d^{1/2}\) are
\[
    m_1
    =
    \frac{1}{d\sqrt{\rho}}
    \left(
        (d-1)\sqrt{2\mu}
        +
        \sqrt{d\lambda+2\mu}
    \right)
\]
on the diagonal, and
\[
    m_2
    =
    \frac{1}{d\sqrt{\rho}}
    \left(
        -\sqrt{2\mu}
        +
        \sqrt{d\lambda+2\mu}
    \right)
\]
off the diagonal. The shear coefficient is
\[
    m_3=\sqrt{\frac{\mu}{\rho}}.
\]

After the above rescaling, Eqs. \eqref{eq:v}, \eqref{eq:sigma_ij}, \eqref{eq:sigma_n matrix} can be written as
\begin{equation}\label{eq:rescaled}
    \begin{split}
    &\partial_t \zeta_i
    =
    m_1\partial_{x_i}v_i
    +
    m_2\sum_{k\neq i}\partial_{x_k}v_k,
    \qquad i=1,\ldots,d,\\
    &\partial_t s_{ij}
    =
    m_3
    \left(
        \partial_{x_j}v_i+\partial_{x_i}v_j
    \right),
    \qquad 1\leq i<j\leq d,\\
    &\partial_t v_i
    =
    m_1\partial_{x_i}\zeta_i
    +
    m_2\sum_{k\neq i}\partial_{x_i}\zeta_k
    +
    m_3\sum_{k\neq i}\partial_{x_k}s_{ik}
    +
    \widetilde f_i,
    \qquad i=1,\ldots,d,
    \end{split}
\end{equation}
where \(s_{ik}=s_{ki}\) if \(k<i\). The source terms
\(\widetilde f_i\) are the corresponding
rescaled source terms.
Now order the rescaled unknowns as
\begin{equation}\label{eq:U}
    U
    =
    (\zeta_1,\ldots,\zeta_d,
    \{s_{ij}\}_{1\leq i<j\leq d},
    v_1,\ldots,v_d)^T. 
\end{equation}
The total number of physical
components is
\(q_d=\frac{d^2+3d}{2}\).
Then the system \eqref{eq:rescaled} can be written in the symmetric matrix form
\begin{equation}\label{eq:dSMF}
    \partial_t U
    =
    \sum_{\ell=1}^d A_\ell \partial_{x_\ell}U
    +
    F.
\end{equation}

\paragraph{Notation for the nested splittings.}
To distinguish the different levels of the circuit decomposition, we adopt the following notation throughout this paper. \(U_{\mathrm I}\) and
\(U_{\mathrm{II}}\) denote the evolutions generated by the
\(-H_1\otimes D_p\) and \(H_2\otimes I\) components, respectively.
The additional subscripts \(x\) and \(Y\) distinguish the two factors obtained
from the splitting of \(U_{\mathrm{II}}\):
\(U_{\mathrm{II},x}\) denotes the spatial-operator evolution, whereas
\(U_{\mathrm{II},Y}\) denotes the evolution generated by the \(Y\)-dependent
term. 
For different boundary conditions, the spatial factor is further denoted by
\(U_P\), \(U_D\), \(U_N\), or \(U_{\mathrm{imp}}\), corresponding to the Dirichlet,
Neumann, and impedance cases, respectively. 

\section{A review of the Schrodingerization method}
\label{sec:schro_review}
The Schr\"odingerisation method embeds a linear,  non-unitary evolution into a
unitary system by introducing one auxiliary variable~\cite{Schrodingerization}.
We recall the construction for autonomous version, as used below.

After spatial discretization, consider
\(\dot{\mathbf y}=B\mathbf y+\mathbf b\), with a time-independent source
\(\mathbf b\). Introduce
\begin{equation}
\mathbf u=\begin{pmatrix}\mathbf y\\\mathbf b\end{pmatrix},
\qquad
A=\begin{pmatrix}B&I\\0&0\end{pmatrix},
\qquad
\dot{\mathbf u}=A\mathbf u.
\label{eq:schro_review_augmentation}
\end{equation}
Then it becomes a homogeneous system. For a homogeneous problem, no augmentation is needed. In either case, write
\begin{equation}
A=H_1+iH_2,
\qquad
H_1=\frac{A+A^\dagger}{2},
\qquad
H_2=\frac{A-A^\dagger}{2i},
\label{eq:schro_review_decomposition}
\end{equation}
where both \(H_1\) and \(H_2\) are Hermitian.
The warped phase transformation \(\mathbf w=g(p)\mathbf u\) gives
\begin{equation}
\partial_t\mathbf w=-H_1\partial_p\mathbf w+iH_2\mathbf w,
\qquad
\mathbf w(0,p)=g(p)\mathbf u(0).
\label{eq:schro_review_transport}
\end{equation}
The original choice is \(g(p)=e^{-|p|}\). With the Fourier convention
\(\widehat{\partial_p\mathbf w}=i\eta\widehat{\mathbf w}\), truncation
and Fourier discretization in \(p\) yield
\begin{equation}
\frac{d}{dt}\widehat{\mathbf w}_h=iH\widehat{\mathbf w}_h,
\qquad
H=-H_1\otimes D_p+H_2\otimes I,
\qquad H=H^\dagger,
\label{eq:schro_review_hamiltonian}
\end{equation}
where \(D_p\) is the Fourier correlation matrix of the variable $p$. Thus the enlarged evolution is \(e^{iHt}\). Depending on the definition
of $\mathbf w$, we can choose either multiplying by $1/g(p)$ at one point of $p$, or integrating to recover $\mathbf u$:
\begin{equation}
    \mathbf u = \mathbf w/g(p), \quad \text{for any} \quad p\geq p^*; \quad \text{or} \quad \mathbf u= \frac{1}{g(p^*)}\int_{p^*}^{\infty} \mathbf w(p)dp,
\end{equation}
where $p^*\geq\max\{\lambda_{\max}^{+}(H_1)T, 0\}$ and \(\lambda_{\max}^{+}(H_1)\) denotes the largest positive eigenvalue of
\(H_1\).

\paragraph{Smooth initialization and optimal auxiliary resolution.}
The cusp of \(e^{-|p|}\) limits the accuracy of the Fourier spectral approximation.   
Instead an analytical extension can lead to optimal matrix-query complexity \cite{matrixqueries}. 
Let $\varepsilon_{\mathrm{aux}}$ denote the error
tolerance assigned to the truncation, recovery, and Fourier
spectral discretization in the auxiliary variable $p$.
Allowing a small
initialization error permits an analytic (Gevrey-1) profile~\cite{matrixqueries}:
\begin{equation}
g_{\varepsilon_{\mathrm{aux}}}(p)
=\frac{1+\operatorname{erf}(\alpha p)}{2}e^{-p},
\qquad
\alpha=2\sqrt{\log(1/\varepsilon_{\mathrm{aux}})}.
\label{eq:schro_review_smooth_profile}
\end{equation}
Here \(0<\varepsilon_{\mathrm{aux}}<e^{-1}\).  
Its relevant quantitative regularity is
\begin{equation}
\left\|g_{\varepsilon_{\mathrm{aux}}}^{(q)}\right\|_{L^2(\mathbb R)}^{1/q}
\leq Cq,
\qquad q=\Theta\!\left(\log(1/\varepsilon_{\mathrm{aux}})\right),
\label{eq:schro_review_derivative_bound}
\end{equation}
with \(C\) independent of \(q\) and \(\varepsilon_{\mathrm{aux}}\).
The Fourier-tail bound \(\mu_{\max}^{-q}\|g^{(q)}\|_{L^2}\), where $\kappa_{\mathrm{aux}}$ is the largest Fourier mode. Together
with controlled domain truncation, then it gives
\begin{equation}
\mu_{\max}=O\!\left(\log(1/\varepsilon_{\mathrm{aux}})\right).
\label{eq:schro_review_resolution}
\end{equation}

\section{Periodic boundary condition}
\label{sec:Periodic}

In this section, we study the elastic wave equation with source term under periodic boundary conditions and develop its explicit quantum circuit representation based on the Schr\"odingerisation framework.

For a homogeneous medium, the one-dimensional SMF with a time independent external force $f$ can be written as
\begin{equation}\label{eq:SMF_1d}
\partial_t u
=
aX\partial_xu+f,
\end{equation}
where \(u = (\sigma/\sqrt{\mu}, \sqrt{\rho}\,v)\) denotes the rescaled variables \eqref{eq:rescaled} in the SMF formulation and \(a=\sqrt{\frac{\mu}{\rho}}\) is the wave speed. 
\subsection{Fourier discretization and Schr\"odingerisation}
Consider the periodic domain $x\in[-\pi L_x,\pi L_x]$, the Fourier spectral method diagonalizes the spatial differentiation operator. Let
\begin{equation}
    x_j=-\pi L_x+j\Delta x,
\qquad
j=0,\ldots,N_x-1,
\qquad
\Delta x=\frac{2\pi L_x}{N_x},
\end{equation}
and let \(\mathbf{u}_h(t)\) and \(\mathbf{f}_h(t)\) collect the grid values of
\(u\) and \(f\), respectively. Define the Fourier differentiation matrix by $P_x=\Phi^{-1}D_x\Phi$, 
where \(\Phi\) is the discrete Fourier transform matrix and
\[
D_x=\sum_{j=0}^{N_x-1}\zeta_j
|j\rangle\langle j|,\qquad \zeta_j=\frac{j-N_x/2}{L_x}
\]
is diagonal in the Fourier basis. Applying the discrete Fourier transform $\widehat{\mathbf{u}} = \Phi\,\mathbf{u}_h$ and $\widehat{\mathbf{f}} = \Phi\,\mathbf{f}_h$, 
the semi-discrete system becomes
\begin{equation}\label{eq:fourier_space_ode}
\frac{d}{dt}\widehat{\mathbf{u}}
=
ia\,X\otimes D_x\,\widehat{\mathbf{u}}
+
\widehat{\mathbf{f}}\triangleq iK_P\,\widehat{\mathbf{u}}
+
\widehat{\mathbf{f}}.
\end{equation}

The presence of the forcing term prevents the semi-discrete system \eqref{eq:fourier_space_ode} from being directly interpreted as a unitary quantum evolution. Following the Schr\"odingerisation method in Section~\ref{sec:schro_review}, we first rewrite the system as a homogeneous ODE 
\begin{equation}\label{eq:homogeneous_system}
\frac{d}{dt}\mathbf{u}
=
A\mathbf{u},\qquad
\mathbf{u}
=
\begin{pmatrix}
\widehat{\mathbf{u}}\\
\widehat{\mathbf{f}}
\end{pmatrix}, \qquad
A=
\begin{pmatrix}
iK_P & I\\
0&0
\end{pmatrix}.
\end{equation}
The Hermitian and
skew-Hermitian parts are
\begin{equation}\label{eq:SMFH1&H2}
H_1=\frac{A+A^\dagger}{2}
=\frac{1}{2}  X\otimes  I^{\otimes n_x+1},
\quad
H_2=\frac{A-A^\dagger}{2i}=\sigma_{00}\otimes K_P+\frac{1}{2}  Y\otimes I^{\otimes n_x+1}.    
\end{equation}
Suppose that the auxiliary variable $p \in [-\pi L_p,\pi L_p]$ and the Schr\"odinger system is
\begin{equation}
    \frac{d }{dt}\widehat{\mathbf w}_h= -i(H_1\otimes D_{p})\widehat{\mathbf w}_h+i(H_2\otimes I)\widehat{\mathbf w}_h\triangleq i H_{SMF}\widehat{\mathbf w}_h,
\end{equation}
where $D_p=\sum_k \eta_k |k\rangle\langle k|$ with diagonal entries $\eta_k=\frac{k-N_p/2}{L_p}$.
The Hamiltonian \(H_{SMF}\) is
\begin{equation}\label{eq:H_SMF}
    H_{SMF} = -H_1\otimes D_{p}+ H_2\otimes I^{\otimes n_p}.
\end{equation}
\begin{remark}
For the \(d\)-dimensional SMF equations, the semi-discrete  system can be written as
\begin{equation}
    \frac{d}{dt}
    \begin{pmatrix}
    \widehat{u}\\
    \widehat{f}
    \end{pmatrix}
    =
    \begin{pmatrix}
    iK_P & I\\
    0 & 0
    \end{pmatrix}
    \begin{pmatrix}
    \widehat{u}\\
    \widehat{f}
    \end{pmatrix},
\end{equation}
where \(K_P=\sum_{\ell=1}^{d} A_\ell\otimes D_\ell .\)
\(D_\ell\) denotes the Fourier differentiation matrix acting on the \(\ell\)-th spatial direction, and \(A_\ell\) is the corresponding symmetric coefficient matrix in \eqref{eq:dSMF}. For the Hamiltonian in Eq.~\eqref{eq:H_SMF}, we have
\begin{equation}\label{eq:H_dSMF}
    H_1
    =
    \frac{1}{2}
    \begin{pmatrix}
    0&I\\
    I&0
    \end{pmatrix}
    =
    \frac{1}{2}X\otimes I,\quad
    H_2
    =
    \frac{1}{2i}
    \begin{pmatrix}
    2iK_P&I\\
    -I&0
    \end{pmatrix}
    =
    \sigma_{00}\otimes K_P
    +\frac{1}{2}Y\otimes I.
\end{equation}
\end{remark}

\subsection{Quantum circuit construction}
Define the Hamiltonian simulation in the short time,
\begin{equation}
    U_{SMF}(\tau) = exp(-i\tau H_1\otimes D_{p}+ i\tau H_2\otimes I^{\otimes n_p} ).
\end{equation}
The total time evolution \(U_{SMF}(T)\) can be approximated via the \(p\)-th-order Trotter product formula.
To explicitly construct the quantum circuit, we decompose  \(U_{SMF}(\tau)\) into a series of small blocks that can be implemented with elementary quantum gates. 
Taking the first-order product formula as an example, for the first splitting, we define
\begin{equation}\label{eq:SMF U_SMF}
U_{SMF}(\tau) \approx U_{\mathrm I}(\tau)U_{\mathrm{II}}(\tau),
\quad
U_{\mathrm I}(\tau) = \exp(-iH_1\otimes D_p \tau),
\quad
U_{\mathrm{II}}(\tau) =  \exp(iH_2\otimes I^{\otimes n_p} \tau)
\end{equation}
Since $D_p$ is diagonal in the Fourier basis, $U_{\mathrm I}(\tau)$ admits an exact circuit implementation without any splitting approximation.
\begin{equation}\label{eq:SMF U_1}
    \begin{aligned}
        U_{\mathrm I}(\tau) &= \exp \left(-\frac{i\tau}{2}  X\otimes I^{\otimes (n_x+1)}\otimes \sum_{k=0}^{N_p-1} \eta_k  \ket{k}\bra{k}\right) = \sum_{k=0}^{N_p-1} \exp (-\frac{i\eta_k\tau}{2}  X\otimes I^{\otimes (n_x+1)})\otimes \ket{k}\bra{k}\\
        &= \sum_{k=0}^{N_p-1}R_x(\tau \eta_k)\otimes I^{\otimes (n_x+1)}\otimes\ket{k}\bra{k} = \sum_{k=0}^{N_p-1}R_x(\frac{\tau}{L_p} )^{k-\frac{N_p}{2}}\otimes I^{\otimes (n_x+1)}\otimes\ket{k}\bra{k}\\
        & = R_x(-\frac{\tau N_p}{2L_p})\otimes I^{\otimes (n_x+1+n_p)}\sideset{}{'}{\prod}_{m=0}^{n_p-1}R_x(\frac{\tau}{L_p})^{2^m}\otimes\ket{1}\bra{1} + I\otimes\ket{0}\bra{0}
        \triangleq V_1(\tau).
    \end{aligned}
\end{equation}
Here the first equality can be found in Appendix \ref{app:pfA_k}. The last step is obtained by applying Lemma 3 in \cite{hu2024quantum} together with the identities \eqref{eq:def of R gate}. The circuit diagram is shown in Fig. \ref{fig:SMF U1}.

For the second splitting, we define
\begin{equation}\label{eq:SMF UII}
    \begin{aligned}
        U_{\mathrm{II}}(\tau)\approx & U_{\mathrm{II},x}(\tau) U_{\mathrm{II},Y}(\tau),\quad
        U_{\mathrm{II},x}(\tau) = \exp (i\tau \sigma_{00} \otimes K_P)\otimes I^{\otimes n_p},\quad 
        U_{\mathrm{II},Y}(\tau) = \exp (\frac{i\tau}{2} Y\otimes I)\otimes I^{\otimes n_p}
        & \\
 & U_{\mathrm{II},x}(\tau) = [\sigma_{00} \otimes U_P(\tau)+\sigma_{11} \otimes I^{\otimes (n_x+1)}]\otimes I^{\otimes n_p},\quad
 U_P(\tau) = \exp(i\tau K_P).
    \end{aligned}
\end{equation}
The operator $U_P(\tau)$ can be treated in exactly the same manner as $U_{\mathrm{I}}(\tau)$, since $D_{x}$ is also diagonal in the Fourier basis. Thus, we have
\begin{equation}
U_P(\tau)=R_x(\frac{a\tau N_x}{ L_x})\otimes I^{\otimes n_x}\sideset{}{'}{\prod}_{m=0}^{n_x-1}(R_x(\frac{-2^{m+1}a \tau}{L_x}) \otimes\ket{1}\bra{1}+I \otimes\ket{0}\bra{0}).
\end{equation}
The circuit diagrams are shown in Figs. \ref{fig:SMF U2x}-\ref{fig:SMF U_2}.
The QFT and inverse QFT are applied to the \(x\)-register and the \(p\)-register before and after the Hamiltonian simulation, respectively. Accordingly, the Schr\"odingerisation-based quantum circuit for solving Eq.~\eqref{eq:SMF_1d} using the first-order Lie--Trotter formula is illustrated in Fig.~\ref{fig:SMF}.

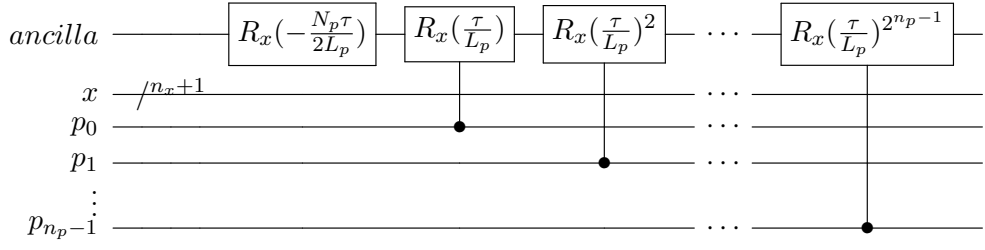
\begin{figure}[h!] 
    \centerline{
        \Qcircuit @C=1em @R=1em
        {
        \lstick{ancilla}&\qw&\qw&\qw&\gate{R_x(-\frac{N_p\tau}{2L_p})}&\gate{R_x(\frac{\tau}{L_p})}&\gate{R_x(\frac{\tau}{L_p})^{2}}&\qw &\cdots& &\gate{R_x(\frac{\tau}{L_p})^{2^{n_p-1}}}&\qw\\
        \lstick{x}&\qw&{/^{n_x+1}}\qw&\qw&\qw&\qw&\qw&\qw&\cdots& &\qw&\qw\\
        \lstick{p_0} &\qw&\qw&\qw&\qw&\ctrl{-2}&\qw&\qw&\cdots& &\qw
        &\qw\\
        \lstick{p_1} &\qw&\qw&\qw&\qw&\qw&\ctrl{-3}&\qw&\cdots&& \qw&\qw\\
        \lstick{\vdots} & & & & & & & & & & \\
        \lstick{p_{n_p-1}} &\qw&\qw&\qw&\qw&\qw&\qw&\qw&\cdots& &\ctrl{-5}&\qw
            }
    }
    \caption{Quantum circuit for $U_{\mathrm{I}}(\tau)$}
    \label{fig:SMF U1}
\end{figure}

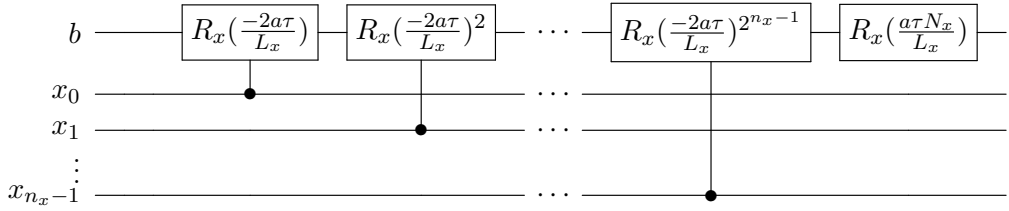
\begin{figure}[h!] 
    \centerline{
        \Qcircuit @C=1em @R=1em
        {
        \lstick{b}&\qw&\qw&\gate{R_x(\frac{-2a\tau}{L_x})}&\gate{R_x(\frac{-2a\tau}{L_x})^2}&\qw &\cdots& &\gate{R_x(\frac{-2a\tau}{L_x})^{2^{n_x-1}}}&\gate{R_x(\frac{a\tau N_x}{L_x})}&\qw\\
        \lstick{x_0} &\qw&\qw&\ctrl{-1}&\qw&\qw&\cdots& &\qw&\qw
        &\qw\\
        \lstick{x_1} &\qw&\qw&\qw&\ctrl{-2}&\qw&\cdots&& \qw&\qw&\qw\\
        \lstick{\vdots} & & & & & & & & & \\
        \lstick{x_{n_x-1}} &\qw&\qw&\qw&\qw&\qw&\cdots& &\ctrl{-4}&\qw&\qw
            }
    }
    \caption{Quantum circuit for $U_{\mathrm{II},x}(\tau)$}
    \label{fig:SMF U2x}
\end{figure}

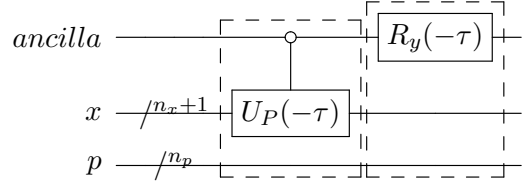
\begin{figure}[h!] 
    \centerline{
        \Qcircuit @C=1em @R=1em
        {
        \lstick{ancilla}&\qw&\qw&\qw&\ctrlo{1}&\gate{R_y(-\tau)}&\qw\\
        \lstick{x} &\qw&{/^{n_x+1} }\qw&\qw&\gate{U_P(-\tau)}&\qw&\qw\\
        \lstick{p}&\qw&{/^{n_p} }\qw&\qw&\qw&\qw&\qw
        \POS"1,5"."2,5"."3,5"!C*+<0.8em>\frm{--}
        \POS"1,6"."2,6"."3,6"!C*+<0.8em>\frm{--}
        }
    }
    \caption{Quantum circuit for $U_{\mathrm{II}}(\tau)$}
    \label{fig:SMF U_2}
\end{figure}

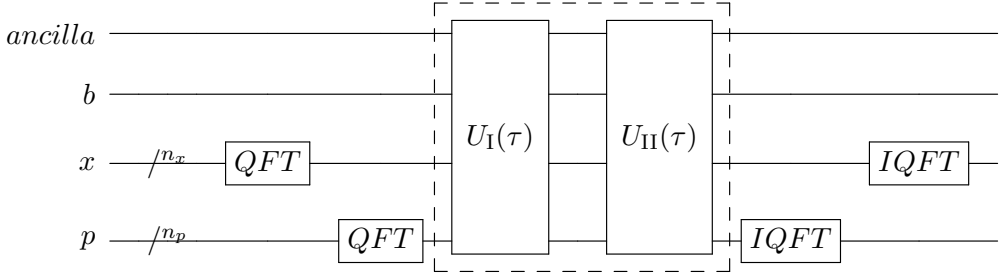
\begin{figure}
    \centerline{
        \Qcircuit @C=1em @R=1.2em
        {
        \lstick{ancilla}&\qw&\qw&\qw&\qw&\qw&\multigate{3}{U_{\mathrm I}(\tau)}&\qw &\multigate{3}{U_{\mathrm{II}}(\tau)} &\qw&\qw&\qw\\
        \lstick{b} &\qw&\qw&\qw&\qw&\qw&\ghost{U_{\mathrm I}(\tau)}&\qw&\ghost{U_{\mathrm{II}}(\tau)} &\qw
        &\qw&\qw\\
        \lstick{x} &\qw&{/^{n_x}}\qw&\qw&\gate{QFT}&\qw&\ghost{U_{\mathrm I}(\tau)}&\qw&\ghost{U_{\mathrm{II}}(\tau)}&\qw& \gate{IQFT}&\qw\\
        \lstick{p} &\qw&{/^{n_p}}\qw&\qw&\qw&\gate{QFT}&\ghost{U_{\mathrm I}(\tau)}&\qw&\ghost{U_{\mathrm{II}}(\tau)}&\gate{IQFT} &\qw&\qw
        \gategroup{1}{7}{4}{9}{1.2em}{--}
            }
    }
    \caption{Quantum circuit for $V_{SMF}(t)$}
    \label{fig:SMF}
\end{figure}
\begin{remark}
In \(d\)-dimensional case, the construction of \(U_{\mathrm I}(\tau)\) and $U_{\mathrm{II},Y}(\tau)$ follows the same procedure as in the one-dimensional case, except that the spatial register now consists of \(n_xd+1+\lceil\log_2 q_d\rceil\) qubits. For \(U_{\mathrm{II},x}\), the circuit implementation of \(\exp(i\tau\sum_{\ell=1}^{d} A_\ell \otimes D_\ell)\) reduces to constructing the individual
evolutions \(e^{iA_\ell t}\), where each \(A_\ell\) is sparse and symmetric. As derived
in Appendix~\ref{sec:complexity}, the implementation of each directional
evolution requires \(O(d\log d)\) CNOT gates.

Fig.~\ref{fig:dSMF U_2x} presents the circuit for $U_{\mathrm{II},x}(\tau)$ in the \(d\)-dimensional case, while the detailed circuits for \(d=2\) and \(d=3\) are  provided in Appendix~\ref{sec:circuit}.    
\end{remark}

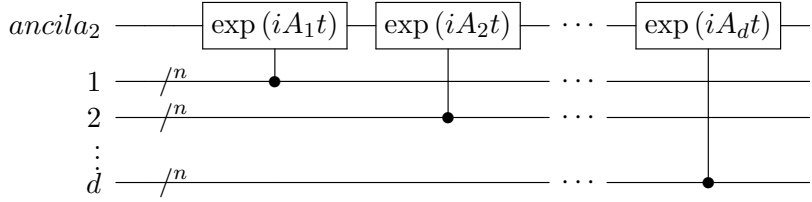
\begin{figure}[h!] 
    \centerline{
        \Qcircuit @C=1em @R=1em
        {
        \lstick{ancila_2}&\qw&\qw&\gate{\exp{(iA_{1}t)}}&\gate{\exp{(iA_{2}t)}}&\qw &\cdots& &\gate{\exp{(iA_{d}t)}}&\qw\\
        \lstick{1} &\qw&{/^{n}}\qw&\ctrl{-1}&\qw&\qw&\cdots& &\qw
        &\qw\\
        \lstick{2} &\qw&{/^{n}}\qw&\qw&\ctrl{-2}&\qw&\cdots&& \qw&\qw\\
        \lstick{\vdots} & & & & & & & & \\
        \lstick{d} &\qw&{/^{n}}\qw&\qw&\qw&\qw&\cdots& &\ctrl{-4}&\qw
            }
    }
    \caption{Quantum circuit for $U_{\mathrm{II},x}(\tau)$}
    \label{fig:dSMF U_2x}
\end{figure}
\subsection{Complexity}

In this subsection, we analyze the approximation error of the Schr\"odingerisation-based circuit simulation.
Let \(U_{SMF}(\tau)\) denote the exact evolution operator over one time step
\(\tau\), and let \(V_{SMF,\,p}(\tau)\) denote its approximation obtained using a
\(p\)th-order product formula. For \(p=1\), the first-order Lie--Trotter
approximation is given by 
\begin{equation}
    V_{SMF,\,1}(\tau) = U_{\mathrm{I}}(\tau) U_{\mathrm{II},x}(\tau) U_{\mathrm{II},Y}(\tau),
\end{equation}
For \(p=2\), we employ the symmetric second-order product formula
\begin{equation}
    \begin{aligned}
V_{SMF,\,2}(\tau)
&=
U_{\mathrm{I}}\!\left(\frac{\tau}{2}\right)
U_{\mathrm{II}}\!\left(\frac{\tau}{2}\right)
U_{\mathrm{II}}\!\left(\frac{\tau}{2}\right)
U_{\mathrm{I}}\!\left(\frac{\tau}{2}\right),
\\
U_{\mathrm{II}}\!\left(\frac{\tau}{2}\right)
& = U_{\mathrm{II},x}\!\left(\frac{\tau}{4}\right) U_{\mathrm{II},Y}\!\left(\frac{\tau}{4}\right)
 U_{\mathrm{II},Y}\!\left(\frac{\tau}{4}\right) U_{\mathrm{II},x}\!\left(\frac{\tau}{4}\right).
    \end{aligned}
\end{equation}

We first present preliminary lemma concerning commutator-scaling estimates for Trotter splitting
errors.

\begin{lemma}\cite{trottererror}
\label{lem:p_order_commutator_error}
Let $G=\sum_{\gamma=1}^{\Gamma}G_\gamma$, 
where each \(G_\gamma\) is anti-Hermitian, and let
\(S_p(\tau)\) be a \(p\)th-order product formula approximating
\(e^{\tau G}\) with stages number $s_p$. Then, the local
approximation error satisfies
\begin{equation}
    \left\|
        e^{\tau G}-S_p(\tau)
    \right\|
    =
    \mathcal{O}_p
    \left(
        \widetilde{\alpha}_{\mathrm{comm}}^{(p)}
        \tau^{p+1}
    \right),
\end{equation}
where
\begin{equation}\label{eq:trotteralpha}
    \widetilde{\alpha}_{\mathrm{comm}}^{(p)}
    =
    \sum_{\gamma_1,\ldots,\gamma_{p+1}=1}^{\Gamma}
    \left\|
    \left[
        G_{\gamma_{p+1}},
        \left[
        G_{\gamma_p},
        \ldots,
        \left[
            G_{\gamma_2},G_{\gamma_1}
        \right]
        \ldots
        \right]
    \right]
    \right\|.
\end{equation}
Here, the constant hidden in \(\mathcal{O}_p(\cdot)\) depends only
on the chosen \(p\)th-order product formula and is independent of
\(\tau\).
For the standard recursive Suzuki
formula, \(s_1=1\) and $s_p=2\cdot5^{p/2-1}, p=2,4,6\ldots$
\end{lemma}

\subsubsection{Error and complexity analysis}
\begin{lemma}\label{lemma:SMF error}
Consider the Schr\"odinger equation $\frac{\partial}{\partial t}\ket{u(t)} = iH_{SMF}\ket{u(t)}$, where the Hamiltonian is given by Eq. \eqref{eq:H_SMF}. 
For a fixed \(p\)th-order product formula, the one-step approximation error
satisfies
\begin{equation}
\left\|
U_{SMF}(\tau)
-
V_{SMF,\,p}(\tau)
\right\|
=
\mathcal{O}_p\!\left(
\tau^{p+1}\mathfrak C_{1,p}
\right),
\label{eq:periodic_local_p_error}
\end{equation}
where
\begin{equation}
\mathfrak C_{1,p}
:=
\kappa_{\mathrm{aux}}(1+\lambda_x)
(1+\kappa_{\mathrm{aux}}+\lambda_x)^{p-1}
+
\lambda_x(1+\lambda_x)^{p-1}.
\label{eq:periodic_commutator_coefficient}
\end{equation}
\(\kappa_{\mathrm{aux}}:=\|D_p\|\) and \(\kappa_x:=\|D_x\|\) denote the largest absolute Fourier wave numbers in the auxiliary \(p\)-direction and the physical \(x\)-direction, respectively, while \(\lambda_x:=a\kappa_x\) denotes the corresponding wave-speed-weighted spatial scale.

For the total evolution time $T=r\tau$, let $V_p(T) = \left[V_{SMF,\,p}(\tau)\right]^r$. The corresponding global approximation error is bounded by 
\begin{equation}\label{}
    \left\|
        U_{SMF}(T)-V_{\mathrm{SMF},p}^r(\tau)
    \right\|=
    \mathcal{O}_p(
        \frac{T^{p+1}}{r^p}\mathfrak C_{1,p}).
\end{equation}
\end{lemma}

\begin{proof}
The error arises from the two splittings. Throughout the proof, we repeatedly use
\begin{equation}
\|[C,D]\|\leq 2\|C\|\,\|D\|.
\label{eq:periodic_basic_commutator_bound}
\end{equation}

For the first splitting in Eq.~\eqref{eq:SMF U_SMF}, its two anti-Hermitian generators are
\begin{equation}
G_1=-iH_1\otimes D_p,
\qquad
G_2=iH_2\otimes I,
\label{eq:periodic_first_splitting_generators}
\end{equation}
where $H_1$ and $H_2$ are given by Eqs.\eqref{eq:SMFH1&H2}. Using the formula \eqref{eq:trotteralpha}, the different innermost commutators satisfy
\begin{equation}
    \begin{aligned}
    [G_1,G_1]&=[G_2,G_2] = 0,\qquad[G_2,G_1]
=-[G_1,G_2]\nonumber\\
[G_1,G_2]
&=[-iH_1,iH_2]\otimes D_p
 =[H_1,H_2]\otimes D_p.
\label{eq:periodic_first_mixed_commutators}
    \end{aligned}
\end{equation}
Using \([X,\sigma_{00}]=-iY\) and \([X,Y]=2iZ\), we have
\begin{equation}
[H_1,H_2]
=-\frac{ia}{2}Y\otimes X\otimes D_x
 +\frac{i}{2}Z\otimes I.
\end{equation}
Consequently,
\begin{equation}
\|[G_1,G_2]\|
=\|[G_2,G_1]\|
\leq
\frac{1}{2}\kappa_{\mathrm{aux}}(1+\lambda_x).
\label{eq:periodic_first_innermost_bound}
\end{equation}
We next determine how each additional commutator changes this initial factor.  For example,
\begin{equation}
    \begin{aligned}
[G_1,[G_1,G_2]]&=[-iH_1,[H_1,H_2]]\otimes D_p^2,\\
[G_2,[G_1,G_2]]
&=[iH_2,[H_1,H_2]]\otimes D_p.
    \end{aligned}
\end{equation}
It follows from \eqref{eq:periodic_basic_commutator_bound} that commuting with \(G_1\) contributes a factor no larger than
\(2\|H_1\|\kappa_{\mathrm{aux}}\), whereas commuting with \(G_2\)
contributes a factor no larger than \(2\|H_2\|\).  Therefore, among the
remaining \(p-1\) commutators, a term containing \(r\) occurrences of
\(G_1\) and \(p-1-r\) occurrences of \(G_2\) is bounded by
\begin{equation}
\bigl(2\|H_1\|\kappa_{\mathrm{aux}}\bigr)^r
\bigl(2\|H_2\|\bigr)^{p-1-r}
\|[G_1,G_2]\|.
\end{equation}
There are \(\binom{p-1}{r}\) arrangements and the two orders of the innermost mixed commutator gives
\begin{align}
\widetilde{\alpha}_1^{(p)}
&\leq
2\sum_{r=0}^{p-1}\binom{p-1}{r}
\bigl(2\|H_1\|\kappa_{\mathrm{aux}}\bigr)^r
\bigl(2\|H_2\|\bigr)^{p-1-r}
\|[H_1,H_2]\|\kappa_{\mathrm{aux}}
\nonumber\\
&=
2^p
\bigl(\|H_1\|\kappa_{\mathrm{aux}}+\|H_2\|\bigr)^{p-1}
\|[H_1,H_2]\|\kappa_{\mathrm{aux}}.
\label{eq:periodic_first_commutator_factor}
\end{align}
Since $\|H_1\|=\frac{1}{2}$, $
\|H_2\|\leq\lambda_x+\frac{1}{2}$ and $
\|[H_1,H_2]\|\leq\frac{1}{2}(1+\lambda_x)$,
Eq.~\eqref{eq:periodic_first_commutator_factor} becomes
\begin{equation}
\widetilde{\alpha}_1^{(p)}
\leq
2^{p-1}\kappa_{\mathrm{aux}}(1+\lambda_x)
(
\frac{1}{2}\kappa_{\mathrm{aux}}+\lambda_x+\frac{1}{2}
)^{p-1}.
\label{eq:periodic_first_commutator_bound}
\end{equation}

For the second splitting in Eq.~\eqref{eq:SMF UII}, using the same method, the two anti-Hermitian generators are $G_1 = ia\sigma_{00}\otimes X\otimes D_x$ and $G_2 = \frac{i}{2}Y\otimes I$. Similar as above, the second commutator coefficient is
\begin{equation}
    \widetilde{\alpha}_2^{(p)}\leq
2^p
\bigl(\|G_1\|+\|G_2\|\bigr)^{p-1}
\|[G_1,G_2]\|=
2^{p-1}\lambda_x
(\lambda_x+\frac{1}{2})^{p-1}.
\label{eq:periodic_second_commutator_bound}
\end{equation}

Finally, applying Lemma~\ref{lem:p_order_commutator_error} to both
splittings and absorbing all fixed formula-dependent factors into
\(\mathcal{O}_p(\cdot)\), the triangle inequality yields
\begin{equation}
    \left\|
U_{SMF}(\tau)-V_{\mathrm{SMF},p}(\tau)
\right\|=
\mathcal{O}_p[
\tau^{p+1}
(
\widetilde{\alpha}_1^{(p)}+
\widetilde{\alpha}_2^{(p)}
)
].
\label{eq:periodic_combined_commutator_error}
\end{equation}
Combining Eqs.~\eqref{eq:periodic_first_commutator_bound} and
\eqref{eq:periodic_second_commutator_bound}, we find
\begin{equation}
\widetilde{\alpha}_1^{(p)}+\widetilde{\alpha}_2^{(p)}
\leq
2^{p-1}\left[
\kappa_{\mathrm{aux}}(1+\lambda_x)
(
\frac{1}{2}\kappa_{\mathrm{aux}}+\lambda_x+\frac{1}{2}
)^{p-1}
+\lambda_x
(\lambda_x+\frac{1}{2})^{p-1}
\right] = \mathcal{O}_p(
\mathfrak C_{1,p}
),
\end{equation}
where in the last step we used
$\frac{1}{2}\kappa_{\mathrm{aux}}+\lambda_x+\frac{1}{2}
\leq 1+\kappa_{\mathrm{aux}}+\lambda_x$ and $\lambda_x+\frac{1}{2}\leq1+\lambda_x$.
Substitution into \eqref{eq:periodic_combined_commutator_error} proves
\eqref{eq:periodic_local_p_error}. For $r$ time steps, we have
\begin{equation}
    \left\|
        U_{SMF}(\tau)-V_{\mathrm{SMF},p}^r(\tau)
    \right\|
    \leq r\left\|
        U_{SMF}(\tau)-V_{\mathrm{SMF},p}(\tau)
    \right\|=
    \mathcal{O}_p(
        \frac{T^{p+1}}{r^p}\mathfrak C_{1,p}).
\end{equation}
\end{proof}

\begin{lemma}\label{lemma:dSMF error}
Let \(V_{d,p}(\tau)\) denote the circuit approximation of the \(d\)-dimensional Hamiltonian evolution \(U_{SMF,\,d}(\tau)\) given by Eq. \eqref{eq:H_dSMF}. Then the local error satisfies
\begin{equation}
\left\|
U_{SMF,\,d}(\tau)-V_{d,p}(\tau)
\right\|
=
\mathcal{O}_p\!\left(
\tau^{p+1}\mathfrak C_{d,p}
\right),
\label{eq:d_periodic_local_p_error}
\end{equation}
where 
\begin{equation}
\mathfrak C_{d,p}
:=
\kappa_{\mathrm{aux}}(1+d\lambda_{\mathrm{sp}})
(1+\kappa_{\mathrm{aux}}+d\lambda_{\mathrm{sp}})^{p-1}
+
d\lambda_{\mathrm{sp}}
(1+d\lambda_{\mathrm{sp}})^{p-1}
+
d^p(d-1)\lambda_{\mathrm{sp}}^{p+1}.
\label{eq:d_periodic_commutator_coefficient}
\end{equation}
For \(T=r\tau\), the global error over the time interval \([0,T]\) satisfies
\begin{equation}
\left\|
U_{SMF,\,d}(T)-[V_{d,p}(\tau)]^{r}
\right\|
=
\mathcal{O}_p\!\left(
\frac{T^{p+1}\mathfrak C_{d,p}}{r^p}
\right),
\label{eq:d_periodic_global_p_error}
\end{equation}
where \begin{equation}
\kappa_{\mathrm{aux}}:=\|D_p\|,
\qquad
\kappa_{\mathrm{sp}}:=\max_{1\leq\ell\leq d}\|D_\ell\|,
\qquad
\nu:=\max_{1\leq\ell\leq d}\|A_\ell\|,
\qquad
\lambda_{\mathrm{sp}}:=\nu\kappa_{\mathrm{sp}}.
\label{eq:d_periodic_wave_numbers}
\end{equation}
\end{lemma}

\begin{proof}
The \(d\)-dimensional construction introduces an additional splitting among the different spatial directions. That is, \(\exp(i\tau\sum_{\ell=1}^{d} A_\ell \otimes D_\ell ) \approx \prod_{\ell=1}^{d}\exp(i\tau A_\ell \otimes D_\ell )\). Define
\begin{equation}
G_p=-iH_1\otimes D_p,
\qquad
G_y=\frac{i}{2}Y\otimes I,
\qquad
G_\ell=i\sigma_{00}\otimes A_\ell\otimes D_\ell,
\quad \ell=1,\ldots,d.
\end{equation}
Set \(G_{\mathrm{sp}}:=\sum_{\ell=1}^dG_\ell\).  From
\eqref{eq:d_periodic_wave_numbers}, we have
\begin{equation}
\|G_p\|=\frac{1}{2}\kappa_{\mathrm{aux}},
\qquad
\|G_y\|=\frac{1}{2},
\qquad
\|G_{\mathrm{sp}}\|\leq d\lambda_{\mathrm{sp}}.
\end{equation}
The splitting-error analysis for the \(d\)-dimensional case follows the same argument as that used in Eqs.~\eqref{eq:periodic_first_commutator_factor} and~\eqref{eq:periodic_second_commutator_bound}, yielding
\begin{equation}
\widetilde{\alpha}_{d,1}^{(p)}
\leq
\kappa_{\mathrm{aux}}(1+d\lambda_x)
\left(
\kappa_{\mathrm{aux}}+d\lambda_x+1
\right)^{p-1},\quad
\widetilde{\alpha}_{d,2}^{(p)}=
d\lambda_x
\left(d\lambda_x+1\right)^{p-1}.
\end{equation}
By Lemma~\ref{lem:p_order_commutator_error} and Eqs.\eqref{eq:periodic_basic_commutator_bound}, each fixed index sequence satisfies
\begin{equation}
\left\|
\left[G_{\gamma_{p+1}},
\left[G_{\gamma_p},\ldots,
\left[G_{\gamma_2},G_{\gamma_1}\right]
\ldots\right]\right]
\right\|
\leq
2^p\lambda_{\mathrm{sp}}^{p+1}.
\end{equation}
There are \(d^{p+1}\) ordered index sequences in total and 
\(d\,d^{p-1}=d^p\) sequences contain the vanishing commutator \([G_{\gamma_i},G_{\gamma_j}]=0\). Thus at most $d^{p+1}-d^{p} = d^p(d-1)$
terms are nonzero, and consequently
\begin{equation}
\widetilde{\alpha}_{\mathrm{d,3}}^{(p)}
\leq
2^p d^p(d-1)\lambda_{\mathrm{sp}}^{p+1}
=
O_p\!\left(d^p(d-1)\lambda_{\mathrm{sp}}^{p+1}\right).
\end{equation}
Combining the three splitting errors by the triangle inequality proves
\eqref{eq:d_periodic_local_p_error}.  Finally, unitary invariance and a
telescoping sum with \(\tau=T/N_t\) give
\eqref{eq:d_periodic_global_p_error}.
\end{proof}
\begin{lemma}\label{lemma:SMF gate}
In one-dimensional case, the approximated time evolution operator $V_{SMF,\,p}(\tau)$ obtained by a $p$th-order Trotter--Suzuki product formula can be implemented using
\begin{equation}
    \mathcal{Q}^{(p)}_{\mathrm{CNOT}}
    =
    \mathcal{O}\!\left(
    2s_p n_p+s_p^2(8n_x+2)
    \right)
    \label{eq:periodic_one_step_gate}
\end{equation}
CNOT gates.
\end{lemma}
\begin{proof}
By Fig. \ref{fig:SMF U1}, one realization of \(U_{\mathrm I}(\tau)\) costs $n_p$ controlled $R_x$ gates, while one
realization of \(U_{\mathrm{II}}(\tau)\) costs
one controlled $R_x$ gate and $n_x$ doubly controlled $R_x$ gates.  
According to \cite{sato2024hamiltonian,vale2023decomposition}, a multi-controlled $R_x$ gate with $(j-1)$ control qubits can be decomposed into at most $(16j-40)$ CNOT gates for $j>3$. Consequently, the implementations of \(U_{\mathrm I}(\tau)\) and \(U_{\mathrm{II}}(\tau)\) require at most
\(2n_p\) and \(8n_x+2\) CNOT gates, respectively.

According to Lemma \ref{lem:p_order_commutator_error}, the \(s_p\)-th-order Trotter product formula consists of \(p\) individual evolution stages. The outer formula contains
\(O(s_p)\) \(U_{\mathrm I}\)-\(U_{\mathrm{II}}\) stages.  Each occurrence
of \(U_{\mathrm{II}}\) is implemented by an inner formula with \(O(s_p)\)
\(U_{\mathrm{II},x}\)-\(U_{\mathrm{II},Y}\) stages.  Thus
\(U_{\mathrm I}\) is repeated \(O(s_p)\) times, whereas the inner blocks
are repeated \(O(s_p^2)\) times. Thus, the total CNOT gates number is \(2s_pn_p+s_p^2(8n_x+2)\).
\end{proof}
\begin{remark}\label{rem:dSMF gate}
In \(d\)-dimensional case,  the circuit consists of
controlled coefficient-matrix evolutions
\(e^{iA_\ell t}\) associated with the \(d\) spatial directions. Each direction
requires \(n_x\) such controlled exponentials, giving \(dn_x\) controlled
exponentials in total.  Since each \(e^{iA_\ell t}\) can be implemented using
\(O(d\log d)\) CNOT gates, which is derived in Appendix \ref{sec:complexity}, and the innermost directional blocks are repeated
\(O(s_p^3)\) times. Thus the approximated time evolution operator $V_{d,p}(\tau)$ obtained by a $p$th-order Trotter--Suzuki product formula can be implemented using $O\!\left(2s_pn_p+16s_p^3n_xd^2\log d\right)$
CNOT gates.
\end{remark}

\begin{lemma}\label{lemma:SMF gate sT}
In one-dimensional case, to achieve
Hamiltonian simulation error between $U_{SMF}(T)$ and $V_{\mathrm{SMF},p}(T)$ at most \(\varepsilon\), it is sufficient to choose time steps $r$ satisfying 
\begin{equation}
r
=
\mathcal{O}_p(
T^{1+1/p}\mathfrak C_{1,p}^{1/p}\varepsilon^{-1/p}
).
\label{eq:periodic_time_steps_p}
\end{equation}
The CNOT complexity of $V_{\mathrm{SMF},\,p}(T)$is
\begin{equation}
\mathcal{Q}^{T,(p)}_{\mathrm{CNOT}}
=
\mathcal{O}_p\!\left[
\bigl(s_p n_p+s_p^2(n_x+1)\bigr)
\frac{T^{1+1/p}\mathfrak C_{1,p}^{1/p}}
{\varepsilon^{1/p}}
\right].
\label{eq:periodic_global_gate_p}
\end{equation}
\end{lemma}

\begin{proof}
To suppress the error in the simulation over the total time $T$ to $\varepsilon$, it is sufficient to divide the total time $T$ into $r$ intervals, where $r = T/\tau$, such that 
\begin{equation}
\left\|
U_{\mathrm{SMF}}(T)-V_{\mathrm{SMF},p}(T)
\right\|
=
\mathcal{O}_p\!\left(
\frac{T^{p+1}\mathfrak C_{1,p}}{N_t^p}
\right)\leq \varepsilon.
\end{equation}
Solving this inequality for \(r\) proves
\eqref{eq:periodic_time_steps_p} and multiplication by the one-step cost in
Lemma~\ref{lemma:SMF gate} proves \eqref{eq:periodic_global_gate_p}.
\end{proof}

\begin{remark}\label{rem:d_periodic_p_gate_sT}
In $d$-dimensional case, to achieve at most \(\varepsilon\) Hamiltonian-simulation
error, it is sufficient to choose $r
=
\mathcal O_p(
T^{1+1/p}\mathfrak C_{d,p}^{1/p}\varepsilon^{-1/p}
)$. 
The corresponding CNOT complexity over \([0,T]\) is $\mathcal Q_{d,\mathrm{CNOT}}^{T,(p)}
=
\mathcal O_p
(s_p^3d^2n_x
T^{1+1/p}\mathfrak C_{d,p}^{1/p}
\varepsilon^{-1/p})$.
\end{remark}

\begin{theorem}\label{thm:1}
Given the one-dimensional SMF equation  \eqref{eq:SMF_1d}, assume that the physical
parameters are independent of the target precision
\(\varepsilon\).  For the auxiliary variable, use the smooth
Schr\"odingerisation profile \eqref{eq:schro_review_smooth_profile}.  Define
\begin{equation}\label{eq:RT}
R_T:=
\max\!\left\{
1,\frac{\|\mathbf u(0)\|}{\|\mathbf u(T)\|}
\right\}.
\end{equation}
Then the \(p\)th-order Schr\"odingerisation circuit prepares
the normalized recovered state
\(\ket{\mathbf u(T)}=\mathbf u(T)/\|\mathbf u(T)\|\) with error
\(O(\varepsilon)\) using
\begin{equation}
\widetilde{\mathcal{O}}_p\!\left(
\frac{s_p^2 T^{1+1/p}R_T^{\,1+1/p}}{\varepsilon^{\,1/p}h}
\right)
\label{eq:periodic_total_complexity_p}
\end{equation}
CNOT gates with polylogarithmic factors suppressed by
\(\widetilde{\mathcal O}_p\).
\end{theorem}

\begin{proof}
The total error consists of the spatial Fourier discretization error, the auxiliary-variable approximation error, and the unitary-evolution error after recovery.

For the spatial discretization, we continue to use \(h\) to denote the spatial resolution. The largest retained wavenumber satisfies $\kappa_x=\|D_x\|=\frac{N_x}{2L_x}
=\frac{\pi}{h}$.
Consequently, 
\begin{equation}\label{eq:SMFx}
    \lambda_x=a\kappa_x
=\frac{\pi a}{h}
=\Theta\left(\frac{a}{h}\right).
\end{equation}

We next estimate the largest Fourier mode in the auxiliary variable. 

Let $P_{N_p}$ denote the Fourier spectral projection in the
auxiliary variable. The standard Fourier approximation
estimate gives
\begin{equation}
\label{eq:fourier_projection_error_SMF}
\left\|
g_{\varepsilon_{\mathrm{aux}}}
-
P_{N_p}g_{\varepsilon_{\mathrm{aux}}}
\right\|_{L^{2}((-L_p,L_p))}
\lesssim
\left(
\frac{\pi}{\kappa_{\mathrm{aux}}}
\right)^{q}
\left\|
g_{\varepsilon_{\mathrm{aux}}}^{(q)}
\right\|_{L^{2}((-L_p,R_p))},
\end{equation}
where $\kappa_{\mathrm{aux}}$ is the largest Fourier mode.
Setting the Fourier projection error to \(O(\varepsilon_{\mathrm{aux}})\) yields
\begin{equation}
\label{eq:mumax_general_SMF}
\kappa_{\mathrm{aux}}
=
O\left(
\varepsilon_{\mathrm{aux}}^{-1/q}
\left\|
g_{\varepsilon_{\mathrm{aux}}}^{(q)}
\right\|_{L^{2}}^{1/q}
\right).
\end{equation}
By the regularity relations \eqref{eq:schro_review_derivative_bound},
Eqs.~\eqref{eq:mumax_general_SMF} imply
\begin{equation}
\label{eq:mumax_estimation_SMF}
\kappa_{\mathrm{aux}}
=
O(
\log (1/{\varepsilon_{\mathrm{aux}}})
)=O(
\log (1/{\varepsilon})
).
\end{equation}
The profile permits that both the truncation and recovery errors are
\(O(\varepsilon_{\mathrm{aux}})\)~\cite{matrixqueries}. In the practical cases considered here, \(d=1,2,3\), so the dominant scales satisfy $\lambda_x > k_p > d$.
Therefore, Eqs.~\eqref{eq:SMFx} and
\eqref{eq:mumax_estimation_SMF}, together with
\eqref{eq:periodic_commutator_coefficient}, imply
$\mathfrak C_{1,p}^{1/p}
=
\widetilde O_p(\frac{a}{h})$.
Moreover, \(n_x=\lceil\log_2N_x\rceil=O(\log(1/h))\), while
\(n_p\) is polylogarithmic in \(1/\varepsilon\).  Hence
\begin{equation}
s_p n_p+s_p^2(n_x+1)
=\widetilde{\mathcal O}_p(s_p^2).
\end{equation}
As shown in Fig.~\ref{fig:SMF}, the QFT on the auxiliary registers are applied before the Hamiltonian simulation and
the inverse QFT applied after it. Their combined CNOT
cost is \(O(n_p^2)\)
\cite{nielsen2010quantum,lin2022lecture} which also can be absorbed into \(\widetilde{\mathcal O}_p\).
Consequently, Lemma~\ref{lemma:SMF gate sT} gives the CNOT cost
of Hamiltonian Simulation with error at most \(\varepsilon\):
\begin{equation}
\mathcal Q_{\mathrm{run}}
=
\widetilde{\mathcal O}_p\!\left(
\frac{s_p^2T^{1+1/p}}
{\varepsilon^{1/p}h}
\right).
\label{eq:periodic_single_run_complexity}
\end{equation}

We next estimate the error introduced by the approximate unitary
evolution after the recovery projection and normalization. 
Let
\begin{equation}
    \ket{\widehat{\mathbf{w}}(T)}
=
U_{\mathrm{SMF}}(T)\ket{\widehat{\mathbf{w}}(0)},
\qquad
\ket{\widehat{\mathbf{w}}_{D}(T)}
=
V_{\mathrm{SMF}}^{r}(\tau)
\ket{\widehat{\mathbf{w}}(0)}
\end{equation}
be the states generated by the exact and approximate unitary
evolutions, respectively. 
And let \(\mathbf w(T)\) and \(\mathbf w_p(T)\) be the corresponding
vectors after the inverse QFT on the auxiliary register. Denote by
\(M_{\mathrm{rec}}\) the projector onto the recovery interval 
\begin{equation}
\label{eq:projected_normalized_error}
\begin{aligned}
&
\left\|
\ket{M_{\mathrm{rec}}\mathbf{w}_{D}(T)}
-
\ket{M_{\mathrm{rec}}\mathbf{w}(T)}
\right\|
=
O\left(
\frac{
\left\|
M_{\mathrm{rec}}
\bigl(
\mathbf{w}_{D}(T)-\mathbf{w}(T)
\bigr)
\right\|
}{
\left\|
M_{\mathrm{rec}}\mathbf{w}(T)
\right\|
}
\right)
\\
&=
O\left(
\frac{
\|\mathbf{w}(0)\|
}{
\|\mathbf{u}(T)\|
\|\boldsymbol g_{\mathrm{rec}}\|
}
\left\|
U_{\mathrm{SMF}}(T)
-
V_{\mathrm{SMF}}^{r}(\tau)
\right\|
\right)
=
O\left(
\frac{\|\mathbf{u}(0)\|}
{\|\mathbf{u}(T)\|}
\varepsilon'
\right),
\end{aligned}
\end{equation}

For the admissible smooth initialization, Schr\"odingerisation recovery
gives~\cite{hu2024quantum,bs,matrixqueries}
\begin{equation}
\|\mathbf w(0)\|
=
\|\mathbf u(0)\|\,\|\boldsymbol g\|,
\qquad
\|M_{\mathrm{rec}}\mathbf w(T)\|
=
\Theta\!\left(
\|\mathbf u(T)\|\,\|\boldsymbol g_{\mathrm{rec}}\|
\right),
\qquad
\frac{\|\boldsymbol g\|}
{\|\boldsymbol g_{\mathrm{rec}}\|}
=O(1),
\end{equation}
where \(\boldsymbol g\) is the discretized auxiliary profile and
\(\boldsymbol g_{\mathrm{rec}}\) is its restriction to the recovery
interval.  
Therefore, to ensure an \(O(\varepsilon)\) error after recovery, we choose $\varepsilon'=\Theta\!\left(\frac{\varepsilon}{R_T}\right)$.
By Eqs. \eqref{eq:periodic_single_run_complexity}, this choice
introduces an additional factor $R_T^{1/p}$
into the CNOT complexity of a single circuit execution. 

Note that the probability of getting the desired state is $\Theta(
\min\left\{
1,
\frac{\|\mathbf u(T)\|^2}
{\|\mathbf u(0)\|^2}
\right\}
)
=
\Theta(R_T^{-2}),$, by using the amplitude amplification, the repeated times for the measurements is $\Theta(R_T)$.
Substituting $\varepsilon'$ into
\eqref{eq:periodic_single_run_complexity} and ignoring the logarithmic term, the proof is completed.
\end{proof}

\begin{remark}
For $\boldsymbol{u}(t)\in H^\rho$, the Fourier projection error satisfies
\[
\|\boldsymbol{u}(t)-P_{N_x}\boldsymbol{u}(t)\|_{L^2}
=O(h^\rho).
\]
Balancing the spatial discretization error with the target accuracy by
choosing $h=\Theta(\varepsilon^{1/\rho})$
gives an overall error of $O(\varepsilon)$. Substituting this choice into
the complexity estimate yields
\begin{equation}
\widetilde{\mathcal{O}}_p\!\left(
\frac{s_p^2 T^{1+1/p}
R_T^{\,1+1/p}}{\varepsilon^{\,1/p+1/\rho}}
\right).
\end{equation}
\end{remark}

\begin{theorem}\label{thm:2}
In $d$-dimensional case, \(d\geq2\), and suppose the assumption same as Theorem~\ref{thm:1}. Use the smooth auxiliary profile
\eqref{eq:schro_review_smooth_profile} and define $R_T$ as Eqs. \eqref{eq:RT}. By  \(p\)th-order Trotter--Suzuki approximation, the Schr\"odingerisation method can prepare the $d$-dimensional SMF solution state with error \(O(\varepsilon)\) using
\begin{equation}
\widetilde{\mathcal O}_p\!\left[
\frac{s_p^3d^{3+1/p}\log(2d)\,T^{1+1/p}R_T^{1+1/p}}
{\varepsilon^{1/p}h^{1+1/p}}
\right].    
\label{eq:d_periodic_total_gate_simplified}
\end{equation}
CNOT gates.
\end{theorem}

\begin{proof}
As in the proof of Theorem~\ref{thm:1}, the spatial discretization and smooth auxiliary profile imply \(\kappa_{\mathrm{sp}}=O(h^{-1})\) and \(\kappa_{\mathrm{aux}}=O(\log(1/\varepsilon))\). 
Hence $\mathfrak C_{d,p}^{1/p} = \widetilde O_p\!\left( d^{1+1/p}h^{-(1+1/p)} \right)$, 
while \(n_x\) and \(n_p\) are polylogarithmic in \(1/\varepsilon\). Ignoring the logarithmic factors \(n_x\) and \(n_p\), Hamiltonian simulation with error \(\varepsilon'\) costs
\begin{equation}
\widetilde{\mathcal O}_p\!\left[
s_p^3d^{3+1/p}\log(2d)\,
T^{1+1/p}
h^{-(1+1/p)}
(\varepsilon')^{-1/p}
\right].    
\end{equation}
The output-state error and recovery measurement are exactly the same as in
the proof of Theorem~\ref{thm:1}. Hence, the overall recovery overhead is $R_T^{1+\frac{1}{p}}$.
\end{proof}

\begin{remark}
\label{rem:classical_complexity_spectral}
For the classical simulation of the $d$-dimensional velocity--stress system with periodic boundary conditions, we consider the same Fourier spectral discretization in space and a $p$th-order partitioned splitting method in time. Let $q_d=\frac{d^2+3d}{2}$
and let $N$ be the number of spatial grid points in each direction. 
One application of the semi-discrete operator requires $\mathcal{O}\!\left(q_d s_d N^d\right)$ arithmetic operations where $s_d$ is the sparsity. Hence, a $p$th-order splitting formula with $s_p$ stages costs
$\mathcal{O}(q_d s_d s_p N^d)$ per time step.

Since the coefficient matrix norm of semi-ODE satisfy
\(\mathcal O(d/h)\), the global temporal
error is $\mathcal O_p(
T^{p+1}d^{p+1}{N_t^{-p}}h^{-p-1})$.
Combining the accuracy requirement with the conservative CFL condition
\(\Delta t=\mathcal O(h/d)\), we choose
\begin{equation}
    N_t=\mathcal O_p\!\left(
\max\left\{
\frac{d T}{h},
\frac{d^{1+1/p}T^{1+1/p}}
{\varepsilon^{1/p}h^{1+1/p}}
\right\}
\right).
\end{equation}
Combining the initial FFT, the classical arithmetic complexity is
\begin{equation}
\mathcal W_{\mathrm{cl}}
=N^d\log N^d+\widetilde{\mathcal O}_p\!\left[
s_pdq_d 
\max\left\{
\frac{d T}{h^{1+d}},
\frac{d^{1+1/p}T^{1+1/p}}
{\varepsilon^{1/p}h^{1+1/p+d}}
\right\}
\right].
\label{eq:d_periodic_partitioned_classical_complexity}
\end{equation}
\end{remark}

\begin{remark}
By Theorems~\ref{thm:1} and~\ref{thm:2}, together with
Remark~\ref{rem:classical_complexity_spectral}, when only considering the quantum state-evolution as a subroutine, for fixed \(d\), the proposed algorithm provides a polynomial
advantage over the classical method with respect to the spatial resolution. 

Moreover, the gate complexity in Theorem~\ref{thm:2} matches that of
Ref.~\cite{Fourierquantum}. Whereas the latter relies on the unitary dynamics
arising from a periodic spectral discretization of the wave equation, the
present Schr\"odingerisation-based construction accommodates source terms and
more general boundary conditions without changing the asymptotic complexity
order.
\end{remark}

\section{Physical boundary conditions}\label{sec:Physical}

In this section, we derive the semi-discrete ODE systems obtained from the staggered-grid finite-difference discretization of the velocity--stress equations under three types of physical boundary conditions. Without loss of generality, we focus on the SMF system \eqref{eq:rescaled} in the subsequent development.

We first consider the one-dimensional velocity--stress formulation,
\begin{equation}
\begin{cases}
\rho(x) v_t(x,t) = \sigma_x(x,t),\quad v(x,0) = v_0(x),\\
\sigma_t(x,t) = \mu(x) v_x(x,t),\quad \sigma(x,0) = \sigma_0(x).
\end{cases}
\label{eq:1d_vs}
\end{equation}

\subsection{The Dirichlet boundary condition}

Taking the domain $\Omega = (a,b)$, we impose the time--independent velocity boundary conditions,
\begin{equation}
v(x=a,t) = g_a,\quad v(x=b,t) = g_b.
\label{eq:velocity_bc}
\end{equation}

We now discretize \eqref{eq:1d_vs} on a staggered grid. The domain is uniformly divied into \(N\) intervals of length \(\Delta x = (b-a)/N\) and \(x_i = a+i\Delta x\) with $i=0,1,\dots,N$. The velocity variables
\(v_i(t)= v(x_i,t)\) are located at integer grid points, while the stress variables
\(\sigma_{i+\frac12}(t) = \sigma(x_{i}+\frac{\Delta x}{2},t)\) are located at half grid points. The staggered finite difference scheme is
given by
\begin{equation}\label{eq:sigma_discrete}
\frac{d}{dt}\sigma_{i+\frac12}(t)
=
a\,\frac{v_{i+1}(t)-v_i(t)}{\Delta x},
\qquad
i=0,1,\ldots,N-1.
\end{equation}
\begin{equation}\label{eq:v_discrete}
\frac{d}{dt}v_i(t)
=
a\,\frac{\sigma_{i+\frac12}(t)-\sigma_{i-\frac12}(t)}{\Delta x},
\qquad
i=1,2,\ldots,N-1,
\end{equation}
Let \(\tilde{\boldsymbol{v}}(t) =
[v_1(t),v_2(t),\ldots,v_{N-1}(t)]^T
\),
\(
\tilde{\boldsymbol{\sigma}}(t)=
[\sigma_{\frac12}(t),\sigma_{\frac32}(t),\ldots,\sigma_{N-\frac12}(t)]^T
\). 
Then \eqref{eq:sigma_discrete}  and \eqref{eq:v_discrete} can be written as
\begin{equation}
\frac{d}{dt}\tilde{\boldsymbol{\sigma}}(t)
=
\frac{a}{\Delta x} D \tilde{\boldsymbol{v}}(t)
+
\frac{a}{\Delta x}\boldsymbol{b}_D,
\label{eq:sigma_matrix}
\end{equation}
\begin{equation}
\frac{d}{dt}\tilde{\boldsymbol{v}}(t)
=
-\frac{a}{\Delta x} D^T\tilde{\boldsymbol{\sigma}}(t),
\label{eq:v_matrix}
\end{equation}
where \(\boldsymbol{b}_D = (-g_a, \dots,g_b)\) contains the contribution from the boundary condition. \(D\) is the difference matrix: 
\[
D=
\begin{pmatrix}
1 &        &        &        \\
-1& 1      &        &        \\
  & -1     & \ddots &        \\
  &        & \ddots & 1      \\
  &        &        & -1
\end{pmatrix}
\in\mathbb{R}^{N\times (N-1)}
\]
To encode the solution into quantum states, we augment the vectors
\(\tilde{\boldsymbol{v}}(t)\) by appending one zero
entry, so that the new variables can be represented as
\[
\boldsymbol{v}(t)
=
\sum_{j=0}^{N-2}v_j(t)\ket{j}+0\ket{N-1},
\qquad
\boldsymbol{\sigma}(t)
=
\sum_{j=0}^{N-1}\sigma_{j+\frac12}(t)\ket{j}.
\]
Here \(N=2^n\) is chosen to match the dimension of the computational basis of an \(n\)-qubit register.
Accordingly, the coefficient matrices are augmented in a compatible manner:
\[A_{\sigma} = [D,\boldsymbol{0}_N],\quad A_{v} = [-D^T;\boldsymbol{0}_N^T].\]
We obtain the new semi-ODEs:
\begin{equation}\label{eq:Dri-discrete_ode}
\frac{d}{dt}
\begin{pmatrix}
\boldsymbol{\sigma}\\
\boldsymbol{v}
\end{pmatrix}
=
\frac{a}{h}
\begin{pmatrix}
0 & A_{\sigma}\\
A_{v} & 0
\end{pmatrix}
\begin{pmatrix}
\boldsymbol{\sigma}\\
\boldsymbol{v}
\end{pmatrix}
+\frac{a}{\Delta x}
\begin{pmatrix}
\boldsymbol{b}_D\\
\boldsymbol{0}
\end{pmatrix}
\triangleq 
\frac{a}{h}(K_D\boldsymbol{u}'+\boldsymbol{b}').
\end{equation}
where 
\[
K_D =
\begin{pmatrix}
0 & A_{\sigma}\\
A_{v} & 0
\end{pmatrix},\quad 
\boldsymbol{u}' = (\boldsymbol{\sigma},\boldsymbol{v})^T,\quad 
\boldsymbol{b}' = (\boldsymbol{b}_D,\boldsymbol{0})^T
\]
Using the shift operators in Appendix \ref{eq:shift_operators} to describe the coefficient matrices \(A_v\) and \(A_{\sigma}\):

\begin{equation}\label{eq:ABP}
    A_{\sigma} = P-S_n^+,\quad A_{v} = -P+S_n^-,\quad  P = I^{\otimes n}-\sigma_{11}^{\otimes n}.
\end{equation}
By homogeneousing, Eqs. \eqref{eq:Dri-discrete_ode} become
\begin{equation}
\frac{d\boldsymbol{u}}{dt}
=
\mathcal{A}_D\boldsymbol{u},
\quad
\boldsymbol{u}
=
\begin{pmatrix}
\boldsymbol{u}'\\
\boldsymbol{b}'
\end{pmatrix},
\quad
\mathcal{A}_D
=\frac{a}{\Delta x}
\begin{pmatrix}
K_D & I\\
0 & 0
\end{pmatrix}.
\label{eq:Dri-semi_ode}
\end{equation}
\subsubsection{Schr\"odingerisation}
The corresponding Hamiltonian after Schr\"odingerisation has the form
\begin{equation}\label{eq:HDri}
H_{Dri}
=
\frac{a}{\Delta x}(-H_1\otimes D_p
+
H_2\otimes I),
\end{equation}
where 
\begin{equation}
H_1
=
\frac{1}{2}X\otimes I, \quad 
H_2
=
-\mathrm{i}\,\sigma_{00}\otimes K_D
+
\frac{1}{2}Y\otimes I.
\label{eq:h2}
\end{equation}

\subsubsection{Interleaved encoding and quantum circuits}
Following the previous notation for the elementary circuits \(U_{\mathrm{I}}(\tau)\), $U_{\mathrm{II}}(\tau)$ and \(U_{\mathrm{II},Y}(\tau)\), we introduce
\[
U_{\mathrm{II},x}(\tau) = \exp(\tau \sigma_{00}\otimes K_D),\quad
U_D(\tau) = \exp(\tau K_D).
\]
For the implementation of \( U_{\mathrm{II},x}(\tau)\), it is equivalent to construct the explicit circuits of \(\exp (\tau K_D)\).
Let \(b\) denote the block qubit distinguishing the stress and
velocity variables, and let \(x\) denote the \(n\)-qubit spatial
register. We use the interleaved register ordering
\begin{equation}
\ket{x}\ket{b},
\end{equation}
where the block qubit is the least significant qubit. This change of
encoding is an exact permutation of the computational basis and
introduces no approximation error.

We employ the technique of interleaved encoding to handle the matrix $K_D$:
\begin{equation}
K_D^{b,x} := K_D
=
\sigma_{01}
\otimes
A_{\sigma}
+
\sigma_{10}
\otimes
A_{v}.
\end{equation}
Exchange the auxiliary qubit register and the $x$ register, and the matrix becomes:
\begin{equation}
\label{eq:K_shift_representation}
K_D^{x,b}
=
S_{n+1}^-
-
S_{n+1}^+
-
\sigma_{11}^{\otimes n}
\otimes
\left(
\sigma_{01}-\sigma_{10}
\right).
\end{equation}
Splitting $K_D^{x,b}$ into two terms, the execution of \(\exp (\tau K_D^{x,b})\) can be derived as follows:
\begin{equation}\label{eq:split3}
     \exp(\tau K_{x,b}) \approx\  \exp(\tau H_0)\exp(\tau H_{end}),
\end{equation}
where
\begin{equation}
    H_0 = S_{n+1}^-
-
S_{n+1}^+,\quad
H_{end} = -
\sigma_{11}^{\otimes n}
\otimes
\left(
\sigma_{01}-\sigma_{10}
\right).
\end{equation}
By Lemma 2 in \cite{hu2024quantum}, the circuit construction of \(\exp(\tau H_0)\) can be implemented by
\begin{equation}\label{eq:split4}
    \exp(\tau H_0)\approx\prod_{j=1}^{n+1}\exp(\tau (s_j^--s_j^+))
    =\prod_{j=1}^{n+1}
I^{\otimes(n_x-j)}
\otimes
W_j\!\left(\gamma\tau,\lambda\right)
\triangleq
V_1^{(4.2)}(\phi_1),
\end{equation}
where \(\gamma = \frac{a}{2h}\) and \(\lambda = \frac{\pi}{2}\).
And \( \exp(\tau H_{end})\) can be implemented by the controlled $R_y$ gate
\begin{equation}
    \exp(\tau H_{end})=\exp(-i\tau \sigma_{11}^{\otimes n}\otimes Y)
    = \sigma_{11}^{\otimes n}\otimes\exp(-i\tau Y)+ \sigma_{00}^{\otimes n}\otimes I = \sigma_{11}^{\otimes n}\otimes R_y( 2\tau).
\end{equation}

\subsubsection{Quantum circuit construction}
The interleaved encoding only changes the register ordering from
$(\text{ancilla},\text{block},x)$ to $(\text{ancilla},x,\text{block})$. It does not affect the construction of \(U_{\mathrm I}(\tau)\) and \(U_{\mathrm{II},Y}(\tau)\) for the identity on the block and \(x\)-registers. The construction of \(U_{\mathrm{II},x}\) is obtained by replacing \(U_P\) with \(U_D\), where \(U_D\) is shown in Figs.\ref{fig:W}-\ref{fig:KD}.
\begin{figure}[h!] 
    \centerline{·
        \Qcircuit @C=1em @R=.7em
        {   \lstick{q_1}& \targ & \qw & \qw & \qw &\qw&\qw&\qw & \ctrl{4} &\qw&\qw &\qw &\qw &\qw &\targ &\qw \\
            \lstick{q_2}& \qw & \targ & \qw & \qw &\qw&\qw&\qw & \ctrl{3} &\qw &\qw &\qw &\qw &\targ &\qw&\qw \\
            \lstick{\vdots}&  &   & \cdots &  && &  & & & & &\cdots&\\
            \lstick{q_{j-1}}& \qw & \qw&\qw  & \targ & \qw &\qw&\qw & \ctrl{1} &\qw &\qw &\targ &\qw &\qw&\qw&\qw\\
            \lstick{q_j}& \ctrl{-4} & \ctrl{-3} &\qw& \ctrl{-1} & \qw & \gate{P(\lambda)} &\gate{H}  &\gate{RZ(-2\gamma \tau)} &\gate{H} & \gate{P(-\lambda)}  & \ctrl{-1}&\qw & \ctrl{-3} & \ctrl{-4} &\qw \\
            \lstick{\vdots}\\
            \lstick{q_{n_x}}& \qw&\qw & \qw & \qw & \qw &\qw&\qw &\qw &\qw &\qw &\qw &\qw &\qw&\qw &\qw}
    }
    \caption{Quantum circuit for $W_{j}(\gamma \tau,\lambda)$}
    \label{fig:W}
\end{figure}
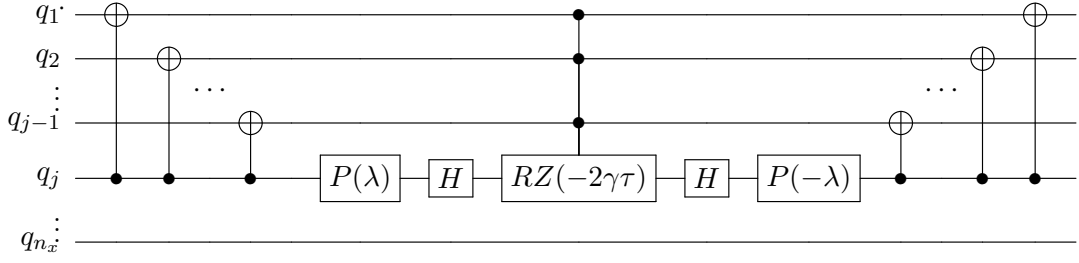
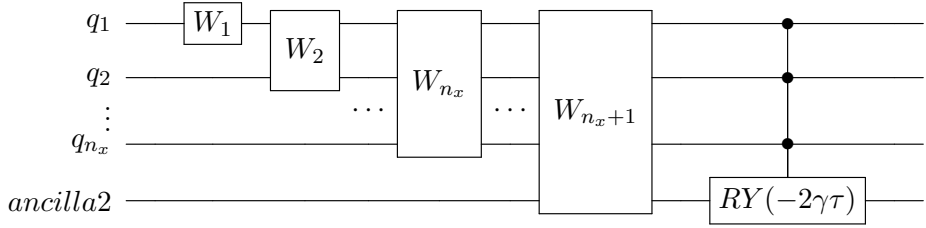
\begin{figure}[h!] 
    \centerline{
        \Qcircuit @C=1em @R=.7em
        {
            \lstick{q_1}
            & \qw 
            & \gate{W_1} 
            & \multigate{1}{W_2}
            & \qw 
            & \multigate{3}{W_{n_x}}
            & \qw 
            & \multigate{4}{W_{n_x+1}} 
            & \qw
            & \ctrl{4}
            & \qw
            & \qw\\
            \lstick{q_2}
            & \qw 
            & \qw
            & \ghost{W_2}
            & \qw 
            & \ghost{W_{n_x}}
            & \qw 
            & \ghost{W_{n_x+1}} 
            & \qw
            & \ctrl{3}
            & \qw
            & \qw\\
            \lstick{\vdots}
            &
            &
            &
            &\cdots
            &
            &\cdots
            &
            &
            &\\
            \lstick{q_{n_x}}
            & \qw 
            & \qw
            & \qw
            & \qw 
            & \ghost{W_{n_x}}
            & \qw 
            & \ghost{W_{n_x+1}} 
            & \qw
            & \ctrl{1}
            & \qw
            & \qw\\
            \lstick{ancilla 2}
            & \qw 
            & \qw
            & \qw
            & \qw 
            & \qw
            & \qw 
            & \ghost{W_{n_x+1}} 
            & \qw
            & \gate{RY(-2\gamma \tau)} 
            & \qw
            & \qw\\
        } 
    }
    \caption{Quantum circuit for $U_D(\tau) = \exp(\tau K_D)$}
    \label{fig:KD}
\end{figure}

\subsection{The Neumann Boundary Condition}\label{subsec:neumann_bc}

We next consider the Neumann boundary condition. In the velocity--stress
formulation, it prescribes the stress at the two endpoints. 
\begin{equation}
\sigma(x_L,t)=g_L,
\qquad
\sigma(x_R,t)=g_R.
\label{eq:neumann_stress_bc}
\end{equation}

In contrast to the Dirichlet case, the stress variables $\sigma(x_i,t)$ are placed at the
integer grid points, while the velocity variables $v(x_{i+\frac12},t)$ are located at the staggered grid points. The staggered grid discretization is given by
\begin{equation}
\frac{\mathrm d}{\mathrm dt}\sigma_i(t)
=
a\frac{v_{i+\frac12}(t)-v_{i-\frac12}(t)}{h},
\qquad
i=1,\ldots,N-1,
\label{eq:neumann_sigma_discrete}
\end{equation}
\begin{equation}
\frac{\mathrm d}{\mathrm dt}v_{i+\frac12}(t)
=
a\frac{\sigma_{i+1}(t)-\sigma_i(t)}{h},
\qquad
i=0,\ldots,N-1.
\label{eq:neumann_v_discrete}
\end{equation}

Following the same procedure as in the Dirichlet case, we set \(N = 2^n\) and append
one zero entry to the interior stress vector, 
\[
\boldsymbol{v}(t)
=
\sum_{j=0}^{N-1}v_{j+\frac12}(t)\ket{j},
\qquad
\boldsymbol{\sigma}(t)
=
\sum_{j=0}^{N-2}\sigma_{j+\frac12}(t)\ket{j}+0\ket{N-1}.
\]
The ODEs can be written as the homogeneous system
\begin{equation}
\frac{\mathrm d}{\mathrm dt}\boldsymbol{u}(t)
=
\mathcal{A}_N\boldsymbol{u}(t),
\qquad
\mathcal{A}_N
=\frac{a}{h}
\begin{pmatrix}
K_N & I\\
0 & 0
\end{pmatrix},\quad \boldsymbol{u} = (\boldsymbol{\sigma},\boldsymbol{v},\boldsymbol{0},\boldsymbol{b}_N),
\label{eq:neumann_homogeneous_ode}
\end{equation}
where \(\boldsymbol{b}_N = (-g_a, \dots,g_b)\) and \(K_N = (X\otimes I)K_D (X\otimes I)\) . 
\subsubsection{Schr\"odingerisation}
The corresponding Hamiltonian after Schr\"odingerisation has the form
\begin{equation}
H_{Neu}
=
\frac{a}{\Delta x}(-H_1\otimes D_p
+
H_2\otimes I),
\label{eq:h_total}
\end{equation}
where
\begin{equation}
H_1
=
\frac{1}{2}X\otimes I, \quad 
H_2
=
-\mathrm{i}\,\sigma_{00}\otimes K_N
+
\frac{1}{2}Y\otimes I.
\end{equation}
\subsubsection{Quantum circuit construction}
The derivation of $U_{Neu}(\tau)$ is the same as $U_{Dir}(\tau)$, except for $\exp(\tau K_N)$. Given the relationship between $K_N$ and $K_D$, we have
\begin{equation}
K^{x,b}_N
=
(I\otimes X)K^{x,b}_D(I\otimes X),\quad
\exp(\tau K^{x,b}_N)
=
(I\otimes X)\exp(\tau K^{x,b}_D)(I\otimes X).
\end{equation}
Therefore, compared with the Dirichlet case, the quantum circuit for the Neumann boundary condition differs only by two additional X gates acting on the auxiliary qubit at both ends of circuit $U_D$, as shown in Fig.~\ref{fig:KN}. The remainder of the circuit is identical.

\begin{figure}[h!] 
    \centerline{
        \Qcircuit @C=1.5em @R=1em
        {
        \lstick{x}
        &{/^{n_x} }\qw
        &\qw
        &\multigate{1}{\exp(\tau K_D)}
        &\qw
        &\qw\\
        \lstick{ancilla} 
        &\qw
        &\gate{X}
        &\ghost{\exp(\tau K_D)}
        &\gate{X}
        &\qw
            }
    }
    \caption{Quantum circuit for $U_N(\tau)=\exp(\tau K_N)$}
    \label{fig:KN}
\end{figure}
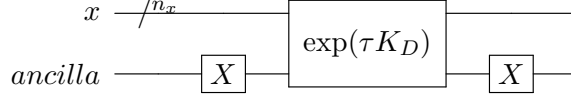

\subsection{The impedance boundary condition}
\label{sec:impedance_schrodingerisation}

The impedance boundary conditions are imposed at the left and right
endpoints as
\begin{equation}
\label{eq:impedance_bc_original}
    -\sigma(0,t)+\eta_L v(0,t)=0,
    \qquad
    \sigma(L,t)+\eta_R v(L,t)=0,
\end{equation}
where \(\eta_L,\eta_R\geq 0\) denote the boundary impedances. Using the same staggered-grid scheme and grid points as in the Dirichlet case, we obtain Eqs.~\eqref{eq:sigma_discrete} and~\eqref{eq:v_discrete}. At the boundary nodes, the corresponding boundary closures are
\begin{equation}
\label{eq:velocity_semidiscrete_imp_boundary}
    \frac{\mathrm d}{\mathrm dt}v_0(t)
    =
    2a\frac{\sigma_{\frac12}(t)
        -
        \eta_Lv_0(t)}{h},
\end{equation}
\begin{equation}
\label{eq:velocity_semidiscrete_imp_boundary_right}
    \frac{\mathrm d}{\mathrm dt}v_N(t)
    =
    -2a\frac{\sigma_{N-\frac12}(t)
        +
        \eta_Rv_N(t)}{h}.
\end{equation}
Let \(N+1=2^n\) and append one zero entry to the interior stress vector. The semi-ODE can be written as
\begin{equation}
\frac{\mathrm d}{\mathrm dt}\boldsymbol{u}(t)
=
\mathcal{A}_I\boldsymbol{u}(t),
\qquad
\mathcal{A}_I
=\frac{a}{h}\begin{pmatrix}
0 & A_{\sigma}\\
A_{v} & -2K
\end{pmatrix},\quad 
\boldsymbol{u} = (\boldsymbol{\sigma},\boldsymbol{v}),
\label{eq:neumann_homogeneous_ode}
\end{equation}
where $K=
    \operatorname{diag}
    \left(
        \kappa_L,0,\ldots,0,\kappa_R
    \right)=
\kappa_L\sigma_{00}^{\otimes n}
+
\kappa_R\sigma_{11}^{\otimes n}$.
The coefficient matrices can be written as
\begin{align}
A_{\sigma} &= -P+S_n^-,
 \quad
A_{v} =P-S_n^++M, \\ 
P &= I^{\otimes n}-\sigma_{11}^{\otimes n},
\quad
M = \sigma_{00}^{\otimes n}-\sigma_{11}^{\otimes(n-1)} \otimes
\sigma_{10}.\label{eq:M}
\end{align}

\subsubsection{Schr\"odingerisation}
The corresponding Hamiltonian after Schr\"odingerisation has the form
\begin{equation}
H_{Imp}
=
\frac{a}{\Delta x}(-H_1\otimes D_p
+
H_2\otimes I),
\label{eq:h_total}
\end{equation}
where
\begin{equation}
H_1
=
\frac{1}{2}
\left(
\sigma_{01}\otimes M^T
+
\sigma_{10}\otimes M
\right)
-
2\sigma_{11}\otimes K,\quad
H_2
=
-\mathrm{i} K_N
+\frac{i}{2}
\left(
\sigma_{01}\otimes M^T
-
\sigma_{10}\otimes M
\right).
\end{equation}

\subsubsection{Quantum circuit construction}
\label{subsec:impedance_boundary_circuit}
\paragraph{Implementation of the \(H_1\)-evolution}
After the register reordering, the Hermitian part \(H_1\) takes the form
\begin{equation}
H_1
=
\frac{1}{2}
\left(
M^T\otimes\sigma_{01}
+
M\otimes\sigma_{10}
\right)
-
2K\otimes\sigma_{11}.
\label{eq:H1_impedance_interleaved}
\end{equation}
Substituting \eqref{eq:M} gives
\begin{align}
H_1
={}&
\sigma_{00}^{\otimes n}
\otimes
(
\frac{1}{2}X
-
2\kappa_L\sigma_{11}
)-
\sigma_{11}^{\otimes(n-1)}
\otimes
(
\frac{1}{2}(
\sigma_{01}^{\otimes 2}
+
\sigma_{10}^{\otimes 2})
+
2\kappa_R
\sigma_{11}^{\otimes 2}
).
\end{align}
We divide \(H_1\) into the left and right boundary terms,
\begin{equation}\label{eq:H1LH1R}
H_{1,L}
=
\sigma_{00}^{\otimes n}
\otimes
h_L,
\quad
H_{1,R}
={}
-\frac{1}{2}
\sigma_{11}^{\otimes(n-1)}
\otimes
\left(
\sigma_{01}^{\otimes 2}
+
\sigma_{10}^{\otimes 2}
\right)-
2\kappa_R
\sigma_{11}^{\otimes n+1},
\end{equation}
where $h_L=-\kappa_L I+\frac{1}{2}X+\kappa_L Z$
and $[H_{1,L},H_{1,R}]=0$.
Hence the decomposition is exact and introduces no splitting error $e^{-i\tau H_1}=
e^{-i\tau H_{1,L}}e^{-i\tau H_{1,R}}$.

For the left boundary term, define $\Omega_L
=
\sqrt{
\kappa_L^2+\frac{1}{4}
}$
and choose \(\theta_L\) such that
\begin{equation}
\cos\theta_L
=
\frac{\kappa_L}{\Omega_L},
\qquad
\sin\theta_L
=
\frac{1}{2\Omega_L}.
\end{equation}
Then $\frac{1}{2}X+\kappa_L Z
=
\Omega_L
R_y(\theta_L)
Z
R_y(-\theta_L)$,
and hence
\begin{equation}
U_L(\tau )
:=
e^{-i\tau h_L}
=
e^{i\kappa_L\tau }
R_y(\theta_L)
R_z(2\Omega_L\tau )
R_y(-\theta_L).
\label{eq:UL_sigma}
\end{equation}
Therefore,
\begin{equation}
e^{-i\tau H_{1,L}}
=
C_{\ket{0^n}}
\left[
U_L^{(b)}(\tau )
\right].
\label{eq:H1L_circuit_sigma}
\end{equation}

For the right boundary term, we define \(V_R=I\otimes \sigma_{00}+X\otimes \sigma_{11}\), which represents a CNOT gate controlled by the qubit \(b\) and acting on the least significant qubit \(x_0\) of the spatial register. Thus, we have
\begin{equation}
V_R
\left(
\sigma_{00}^{(x_0)}
\otimes X_b
\right)
V_R
=
\sigma_{01}^{(x_0)}
\otimes\sigma_{01}^{(b)}
+
\sigma_{10}^{(x_0)}
\otimes\sigma_{10}^{(b)},\quad
V_R
\left(
\sigma_{00}^{(x_0)}
\otimes\sigma_{11}^{(b)}
\right)
V_R
=
\sigma_{11}^{(x_0)}
\otimes\sigma_{11}^{(b)}.
\end{equation}
The right boundary Hamiltonian\eqref{eq:H1LH1R} can be written as
$H_{1,R}
=
V_R
\left[
\sigma_{11}^{\otimes(n-1)}
\otimes
\sigma_{00}
\otimes
h_R
\right]
V_R$ 
with $h_R
=
-\kappa_R I
-
\frac{1}{2}X
+
\kappa_R Z$. Similarly, 
\begin{equation}
    U_R(\tau )
:=
e^{-i\tau h_R}
=
e^{i\kappa_R\tau }
R_y(\theta_R)
R_z(2\Omega_R\tau )
R_y(-\theta_R)
\end{equation}
where $\Omega_R
=
\sqrt{
\kappa_R^2+\frac{1}{4}
}$ 
and \(\theta_R =\operatorname{atan2}(-1/2,\kappa_R) \).
\begin{equation}
e^{-i\tau H_{1,R}}
=
V_R
C_{\ket{1^{n-1}0}}U_R^{(b)}(\tau )
V_R.
\label{eq:H1R_circuit_sigma}
\end{equation}

Combining \eqref{eq:H1L_circuit_sigma} and
\eqref{eq:H1R_circuit_sigma}, we obtain
\begin{equation}
e^{-i\tau H_1}
=
C_{\ket{0^{n}}}
\left[
U_L^{(b)}(\tau )
\right]
V_R
C_{\ket{1^{n-1}0}}
\left[
U_R^{(b)}(\tau )
\right]
V_R.
\label{eq:H1_complete_sigma}
\end{equation}
The corresponding circuit is shown in Fig.~\ref{fig:Imp U1}.

\paragraph{Implementation of the \(H_2\)-evolution}
For the evolution generated by \(H_2\), the \(K_N\)-dependent part is implemented by the circuit shown in Fig.~\ref{fig:KN}. We denote the remaining part by \(H_{2,M}\). Following the same treatment as for \(H_1\), its evolution can be written as
\begin{equation}
e^{i\tau H_{2,M}}
=
C_{\ket{0^{n}}}
\left[R_y^{(b)}(\tau)\right]
V_R C_{\ket{1^{n-1}0}}
\left[ R_y^{(b)}(-\tau) \right] V_R.
\label{eq:H2M_complete_sigma}
\end{equation}
The corresponding circuit is shown in Fig.~\ref{fig:Imp U_2}.
\subsubsection{Circuits Plots}
\begin{figure}[htbp]
\centering
\centerline{
\Qcircuit @C=1.0em @R=1em{
\lstick{x}
    & {/^{n_x-1}}\qw
    & \ctrlo{1}
    & \ctrlo{1}
    & \ctrlo{1}
    & \ctrlo{1}
    & \qw
    & \ctrl{1}
    & \ctrl{1}
    & \ctrl{1}
    & \ctrl{1}
    & \qw
    & \qw
\\
\lstick{x_0}
    & {/^{1}}\qw
    & \ctrlo{1}
    & \ctrlo{1}
    & \ctrlo{1}
    & \ctrlo{1}
    & \targ
    & \ctrlo{1}
    & \ctrlo{1}
    & \ctrlo{1}
    & \ctrlo{1}
    & \targ
    & \qw
\\
\lstick{ancilla}
    & \qw
    & \gate{P(\kappa_L\tau)}
    & \gate{R_y(-\theta_L)}
    & \gate{R_z(2\Omega_L \tau)}
    & \gate{R_y(\theta_L)}
    & \ctrl{-1}
    & \gate{P(\kappa_L\tau)}
    & \gate{R_y(-\theta_R)}
    & \gate{R_z(2\Omega_R \tau)}
    & \gate{R_y(\theta_R)}
    & \ctrl{-1}
    & \qw
\gategroup{2}{3}{3}{6}{0.8em}{--}
\POS"2,7"."3,7"."2,12"."3,12"
    ."3,8"."3,9"."3,10"."3,11"
    !C*+<0.8em>\frm{--}
}
}
\caption{Quantum circuit for $U_{\mathrm{I},x}(\tau)$}
\label{fig:Imp U1x}
\end{figure}
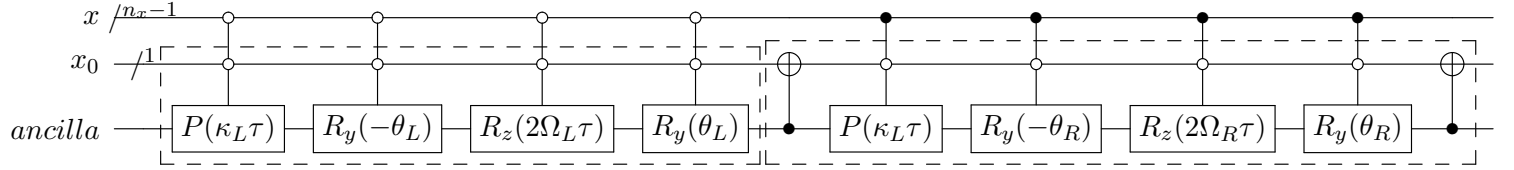

\begin{figure}[h!] 
    \centerline{
        \Qcircuit @C=1em @R=1em
        {
        \lstick{ancilla}&\qw&\qw&\gate{U_{\mathrm{I},x}(\tau)}&\gate{U_{\mathrm{I},x}^2(\tau)}&\qw &\cdots& &\gate{U_{\mathrm{I},x}^{2^{n_x-1}}(\tau)}&\qw\\
        \lstick{q_0} &\qw&\qw&\ctrl{-1}&\qw&\qw&\cdots& &\qw
        &\qw\\
        \lstick{q_1} &\qw&\qw&\qw&\ctrl{-2}&\qw&\cdots&& \qw&\qw\\
        \lstick{\vdots} & & & & & & & & \\
        \lstick{q_{n_x-1}} &\qw&\qw&\qw&\qw&\qw&\cdots& &\ctrl{-4}&\qw
            }
    }
    \caption{Quantum circuit for $U_{\mathrm{I}}(\tau)$}
    \label{fig:Imp U1}
\end{figure}
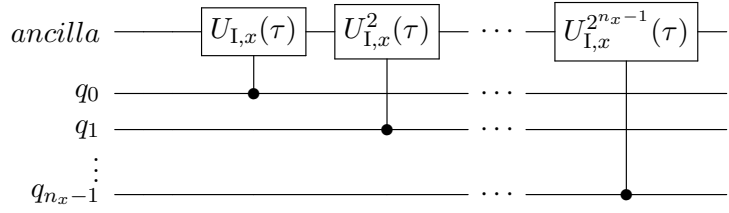

\begin{figure}[h!] 
    \centerline{
        \Qcircuit @C=1em @R=1em
        {
        \lstick{x}
        &{/^{n_x-1} }\qw
        &\qw
        &\multigate{2}{\exp(\tau K_D)}
        &\qw
        &\ctrlo{1}
        &\qw
        &\ctrl{1}
        &\qw
        &\qw
        \\
        \lstick{x_0}
        &{/^{1} }\qw
        &\qw
        &\ghost{\exp(\tau K_D)}
        &\qw
        &\ctrlo{1}
        &\targ
        &\ctrlo{1}
        &\targ
        &\qw\\
        \lstick{ancilla} 
        &\qw
        &\gate{X}
        &\ghost{\exp(\tau K_D)}
        &\gate{X}
        &\gate{R_y(-\tau)}
        &\ctrl{-1}
        &\gate{R_y(\tau)}
        &\ctrl{-1}
        &\qw
\gategroup{1}{6}{3}{9}{0.8em}{--}
}
}
    \caption{Quantum circuit for $U_{\mathrm{II}}(\tau)$}
    \label{fig:Imp U_2}
\end{figure}
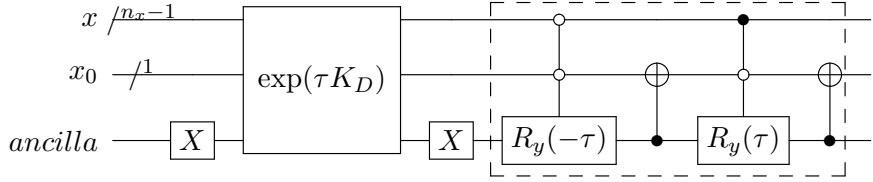

\subsection{Complexity}

For conciseness, the following error, gate-complexity, and total-complexity
analyses are presented only for the Dirichlet boundary condition. The Neumann
and impedance constructions have the same asymptotic estimates, up to constants
depending on the boundary parameters, and hence their conclusions follow
analogously.

\subsubsection{Splitting errors}
\begin{lemma} \label{lem:p_order_trotter} 
Consider the Schr\"odinger equation $\frac{\partial}{\partial t}\ket{u(t)} = iH_{\mathrm{Dri}}\ket{u(t)}$, where the Hamiltonian is given by Eq. \eqref{eq:HDri}. Let $U_{\mathrm{Dri}}(T) = e^{iH_{\mathrm{Dri}}T}$ denote the exact time-evolution operator. For a time step $\tau=T/r$, let $V_{\mathrm{Dri},\,p}(T)$ be a $p$th-order Trotter--Suzuki product formula. Then the local approximation error over one time step satisfies 
the local splitting error satisfies
\begin{equation}
\label{eq:p_order_local_error}
    \left\|
        U_{\mathrm{Dri}}(\tau)
        -
        V_{\mathrm{Dri},\,p}(\tau)
    \right\|
    =
    \mathcal{O}_p\left(
            \frac{a^{p+1}\tau^{p+1}}{h^{p+1}}
            (k_{\max}^{p}
            +
            n^{p})\right).
\end{equation}
For the total evolution time $T=r\tau$, let $V_{\mathrm{Dri},\,p}(T) = V^r_{\mathrm{Dri},\,p}(\tau)$. The corresponding global approximation error is bounded by 
\begin{equation}
\label{eq:p_order_global_error}
    \left\|
        U_{\mathrm{Dri}}(T)- V^r_{\mathrm{Dri},\,p}(\tau)
    \right\|
    =
    \mathcal{O}_p\left(
        \frac{a^{p+1}T^{p+1}}{r^ph^{p+1}}
            (k_{\max}^{p}
            +
            n^{p})\right).
\end{equation}
\end{lemma}
\begin{proof}
The error analysis for the first two splitting steps is the same as that in Theorem~\ref{thm:1}. Since $\|H_1\|=\frac{1}{2}$, $\|H_2\|\leq\|K_D\|+\frac{1}{2}$ and $\|K_D\|\leq 2\|$, we have
\begin{equation}\label{eq:dirichlet_error_1_p}
    \begin{aligned}
    \widetilde{\alpha}_{1}^{(p)}
    &\leq
    \frac{2^pa^{p+1}}{h^{p+1}}
    k_{\max}
    \|[H_1,H_2]\|
    \left(
        k_{\max}\|H_1\|
        +
        \|H_2\|
    \right)^{p-1}\\
    &=
    \mathcal{O}_p
    \left(
        \frac{a^{p+1}}{h^{p+1}}
        k_{\max}(1+k_{\max})^{p-1}
    \right)
    =
    \mathcal{O}_p
        ( \frac{a^{p+1}}{h^{p+1}}
        k_{\max}^{p}).          
    \end{aligned}
\end{equation}
\begin{equation}
    \widetilde{\alpha}_{2}^{(p)}
    =
    \mathcal{O}_p
    \left(
        \frac{a^{p+1}}{h^{p+1}}
        \|K_D\|
        (\|K_D\|+\frac{1}{2})^{p-1}
    \right)
    =
    \mathcal{O}_p
    \left(
        \frac{a^{p+1}}{h^{p+1}}
    \right).
    \label{eq:dirichlet_error_2_p}
\end{equation}

The thrid spliiting is \eqref{eq:split3}, by Appendix B.1 in \cite{hu2024quantum}, using $\|[H_0,H_{end}]\|\leq1$, $\|H_0\|\leq2$ and $\|H_{end}\|=1$,
$\widetilde{\alpha}_{3}^{(p)}$ immediately yields the same order as $\widetilde{\alpha}_{2}^{(p)}$.

It remains to consider the decomposition of the main shift term \eqref{eq:split4}.
Write $T_\ell=s_\ell^{-}-s_\ell^{+}$ and  define the commutator coefficient by
\begin{equation}
    \alpha_{4}^{(p)}
    :=
    \sum_{j_1,\ldots,j_{p+1}=1}^{n+1}
    \left\|
    \left[
    T_{j_{p+1}},
    \left[
    \cdots,
    \left[
    T_{j_2},T_{j_1}
    \right]
    \cdots
    \right]
    \right]
    \right\|.
\end{equation}
By Appendix (B.6) in \cite{hu2024quantum}, $ \sum_{1\leq j<k\leq n+1}
    \|[T_j,T_k]\|
    =
    \mathcal{O}(n)$ and \(\|T_j\|\leq1\).
Furthermore, the conservative higher-order is estimated
$\mathcal{O}_p(n^p)$.
Consequently,
\begin{equation}
    \alpha_{4}^{(p)}
    =
    \mathcal{O}_p
        \left(\frac{n^p a^{p+1}}{h^{p+1}}\right).
    \label{eq:dirichlet_error_4_p}
\end{equation}

Therefore, by the triangle inequality, the total local error is bounded by the sum of the four contributions,
\begin{equation}
    \epsilon_{Dri}^{(p)}
    =
    \mathcal{O}_p
    \left[\frac{ a^{p+1} \tau^{p+1}}{h^{p+1}}
            (k_{\max}(1+k_{\max})^{p-1}+1+n^p)
    \right].
    \label{eq:dirichlet_error_bound_p}
\end{equation}
By the same argument as in Lemma~\ref{lemma:SMF error}, the error over time \(T\) is bounded by Eq.~\eqref{eq:p_order_global_error}.
\end{proof}

\begin{remark}
\label{rem:d_dimensional_error}
In the \(d\)-dimensional case, \(K_D\) is composed of the differential matrix associated with different spatial directions. As in Lemma~\ref{lemma:dSMF error}, an additional Trotter--Suzuki splitting is required to implement the evolution generated by \(K_D\). The derivation is analogous to that of Lemma~\ref{lemma:dSMF error} and is omitted here. The local approximation error for the \(p\)th-order product formula satisfies
\begin{equation}
    \epsilon_d^{(p)}
    =
    \mathcal{O}_p
    \left(
        \frac{a^{p+1}\tau^{p+1}}{h^{p+1}}
            (dk_{\max}(d+k_{\max})^{p-1}+d^{p+1}+dn^p)
    \right).
    \label{eq:dirichlet_error_bound_p}
\end{equation}
Therefore, for a total evolution time \(T=r\tau\), the corresponding
global error is
\begin{equation}
    \left\|
        U_{\mathrm{Dri}}^{(d)}(T)
        -
        V_{Dri,p}^{(d)}(T)
    \right\|
    =
    \mathcal{O}_p
    \left(
         \frac{a^{p+1}T^{p+1}}{r^ph^{p+1}}
         \left(dk_{\max}(d+k_{\max})^{p-1}+d^{p+1}+dn^p\right)
    \right).
\end{equation}
\end{remark}

\subsubsection{CNOT-gate complexity}

\begin{lemma}\label{lem:gate_count_dirichlet}
    The approximated time evolution operator $V_{{Dri},\,p}(\tau)$ obtained by a $p$th-order Trotter--Suzuki product formula can be implemented using $\mathcal{Q}_{{\text{Dri}}} = O(s_p^4n^2)$ CNOT gates.
\end{lemma}
\begin{proof}
We first collect the costs of the elementary circuit blocks. $U_{\mathrm I}(\tau)$ contains one \(R_x\) gate
and \(n_p\) controlled \(R_x\) gates. Hence its CNOT cost is $2n_p$.
$U_{\mathrm{II},Y}(\tau)$ contains one \((n+1)\)-qubit
controlled \(R_y\) gate with CNOT cost $16n-8$.
The controlled \(W_j\) blocks contribute
\begin{align}
    C_{\mathrm{W}}
    &=
    \sum_{j=2}^{n+1}(16j-24)
    +
    6\sum_{j=2}^{n+1}2(j-1)
    +
    4(2+2n)
    \nonumber\\
    &=
    14n^2+14n+8.
    \label{eq:one_shift_sweep_cost}
\end{align}
Since the circuit construction involves four nested
product-formula decompositions, the resulting gate cost acquires 
\begin{equation}
    N_{\mathrm{CNOT}}
    \leq
    2s_p n_p
    +
    s_p^3(16n-8)
    +
    s_p^4(14n^2+14n+8).
    \label{eq:p_order_dirichlet_cnot}
\end{equation}
Hence, the dominant gate complexity is
$s_p^4 C_{\mathrm W}$.
\end{proof}

\begin{remark}
\label{rem:d_dimensional_gate}
In the \(d\)-dimensional case, the gate complexity of $U_{\mathrm{II}}(\tau)$ increases by a factor of \(d\) compared with the one-dimensional case. Then the shift-operator
circuits are repeated \(s_p^5\) times in a \(p\)th-order Trotter--Suzuki approximation. Consequently, the
CNOT complexity satisfies $\mathcal{O}(d\,s_p^5 n^2)$.
\end{remark}

\begin{lemma}\label{lemma:Dri:complexity:T}
In one-dimensional case, to achieve
Hamiltonian simulation error between $U_{Dri}(T)$ and $V_{\mathrm{Dri},p}(T)$ at most \(\varepsilon\), it is sufficient to choose time steps $r$ satisfying $\mathcal{O}_p(\frac{T^{1+\frac{1}{p}}(k_{\max}^{p}+
            n^{p})^{\frac{1}{p}}}{{ \varepsilon^{\frac{1}{p}}}h^{1+\frac{1}{p}}})$.
The CNOT complexity of $V_{\mathrm{Dri},\,p}(T)$ is
$\mathcal{Q}^{T,(p)}_{\mathrm{CNOT}}
=
\mathcal{O}_p\!\left(
\frac{s_p^4n^2T^{1+\frac{1}{p}}(k_{\max}^{p}+n^{p})^{\frac{1}{p}}}{{ \varepsilon^{\frac{1}{p}}}h^{1+\frac{1}{p}}}\right)$.
\end{lemma}

\begin{proof}
The proof is analogous to that of Lemma~\ref{lemma:SMF gate sT} and is omitted here.
\end{proof}

\begin{remark}\label{rem:dDDri:complexity:T}
In $d$-dimensional case, to achieve at most \(\varepsilon\) Hamiltonian-simulation
error,
the corresponding CNOT complexity over \([0,T]\) is $O\left(\frac{d^{2+\frac{1}{p}}s_p^5n^2T^{1+\frac{1}{p}}}{{ \varepsilon^{\frac{1}{p}}}h^{1+\frac{1}{p}}}\right)$.
\end{remark}

\subsubsection{Total complexity and classical comparison}
\begin{theorem}
\label{thm:3}
    In one-dimensional case, suppose the assumption same as Theorem~\ref{thm:1}. Use the smooth auxiliary profile    \eqref{eq:schro_review_smooth_profile} and define $R_T$ as Eqs. \eqref{eq:RT}. By  \(p\)th-order Trotter--Suzuki approximation, the CNOT gate complexity of the Schr\"odingerisation-based quantum algorithm for the SMF equation \eqref{eq:velocity_bc} with Dirichlet boundary conditions is
$$\widetilde{O}\left(
\frac{s_p^4T^{1+\frac{1}{p}}R_T^{1+\frac{1}{p}}}
{\varepsilon^{\frac{1}{p}}
h^{1+\frac{1}{p}}}
\right).$$
\end{theorem}

\begin{proof}
    The proof follows the same steps as those of Theorems~\ref{thm:1} and~\ref{thm:2} and is therefore omitted.
\end{proof}

\begin{remark}
According to Remark~\ref{rem:dDDri:complexity:T} and Theorem~\ref{thm:3}, in the \(d\)-dimensional case, the CNOT complexity of the circuit implementation for solving the SMF with Dirichlet boundary conditions via Schr\"odingerisation is 
$\widetilde{O}\left(
\frac{d^{2+\frac{1}{p}}T^{1+\frac{1}{p}}
s_p^5R_T^{1+\frac{1}{p}}
}{
\varepsilon^{\frac{1}{p}}
h^{1+\frac{1}{p}}
}
\right)$.
\end{remark}

\begin{remark}\label{rem:classical_complexity_staggered}
For the \(d\)-dimensional velocity--stress system discretized spatially by the
same staggered-grid scheme and temporally by a \(p\)th-order operator-splitting
scheme, the classical arithmetic complexity can be derived analogously to
Remark~\ref{rem:classical_complexity_spectral} and is given by 
\begin{equation}
\mathcal W_{\mathrm{cl}}
=\mathcal O_p\!\left[
s_pdq_ds_d
\max\left\{
\frac{d T}{h^{1+d}},
\frac{d^{1+1/p}T^{1+1/p}}
{\varepsilon^{1/p}h^{1+1/p+d}}
\right\}
\right]
\label{eq:d_staggered_classical_complexity}
\end{equation}
\end{remark}

\begin{remark}
For fixed \(d\), and when only the quantum state-evolution subroutine is
considered, the staggered-grid Schr\"odingerisation algorithm retains a
{\it polynomial} advantage in the spatial resolution over the corresponding
classical method. Ignoring the logarithmic factors, the gate complexity has the same order as that of Ref.~\cite{Sato3Delastic} which considers a centered finite-difference discretization under periodic boundary conditions. The present method accommodates Dirichlet, Neumann, and impedance boundary conditions. 

Moreover, the symmetric structure of the SMF further simplifies the circuit decomposition. In particular, the implementation of
\(U_{\mathrm{I}}\) requires \(O(n_p)\) CNOT gates, where
\(n_p=\log_2 N_p\), rather than a cost depending directly on the auxiliary
discretization size \(N_p\)\cite{hu2024quantum}. Thus, the proposed construction reduces the
circuit-level overhead while preserving the same leading asymptotic complexity.
\end{remark}

\section{Numerical Experiments}\label{sec:Experiments}
In this chapter, we present numerical experiments to verify the accuracy of the derived quantum circuits. 
For each test, we compare four types of solutions: the exact
PDE solution, a first-order classical Lie-partitioned solution of the semidiscrete system, the classical solution of the Schr\"odingerised system, and the state recovered from the UnitaryLab circuit.  The latter two use the same auxiliary discretization and recovery point.  All circuit results in this section are obtained with a noiseless state-vector simulator.  Both tests use $\rho=1$, with $\mu=1/2$ in the periodic test and $\mu=1$ in the impedance test.
\subsection{A two-dimensional periodic problem with a constant body force}

On $[-\pi,\pi)^2$, we consider the two-dimensional velocity--stress system
with periodic boundary conditions in both directions and the constant body force
$\mathbf f=(0.2,-0.1)^T$.

We use $n_x=n_y=4$ spatial qubits and
$N_t=512$ time steps up to $T=1$. The classical method updates stress first and
then velocity, whereas the circuit uses first-order splitting of the
Schr\"odingerised evolution, including the two spatial directions.
For these material parameters, the SMF variables are
$(\sigma_{11},\sigma_{22},\sqrt{2}\sigma_{12},v_x,v_y)^T$.
The five components are encoded in three qubits, with three unused basis
states padded by zeros.

For Schr\"odingerisation, we set $n_p=7$, $R=4$, and use the smooth
auxiliary profile \eqref{eq:schro_review_smooth_profile}. 
Including the source-augmentation qubit, the circuit contains
$4+4+3+1+7=19$ qubits. Fig.~\ref{fig:numerical_periodic_force} compares the five physical
components along $y=0$, and Fig.~\ref{fig:numerical_periodic_fields}
shows the full two-dimensional fields. The four grid solutions are
Fourier-interpolated to a $128\times128$ grid for visualization only.

\begin{figure}[h]
    \centering
    \includegraphics[width=0.98\linewidth]{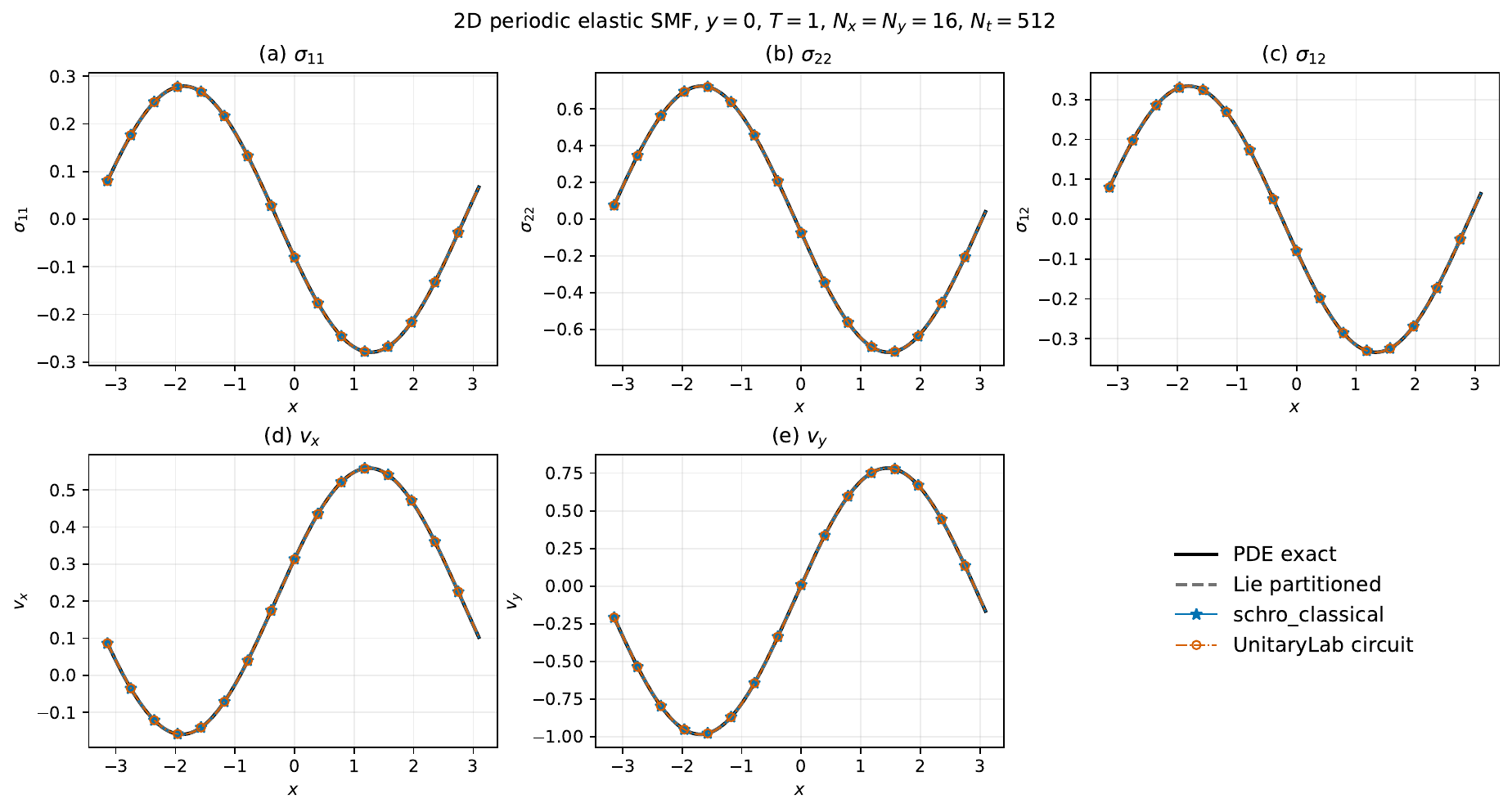}
    \caption{Five components of the two-dimensional 
    velocity--stress solution along $y=0$ at $T=1$.}
    \label{fig:numerical_periodic_force}
\end{figure}

\begin{figure}[h]
    \centering
    \includegraphics[width=0.98\linewidth]{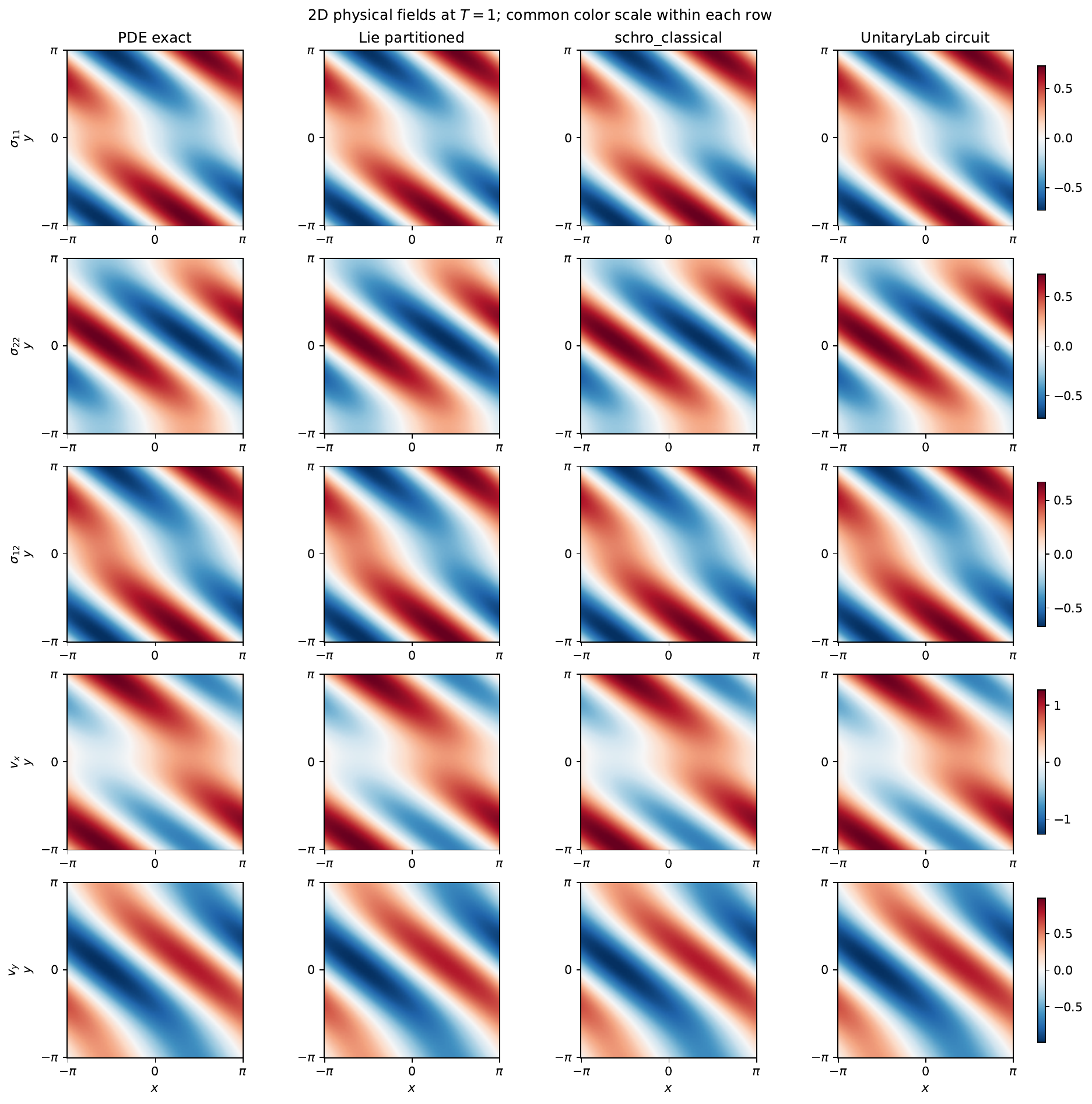}
    \caption{Two-dimensional physical fields at $T=1$ on the
    $16\times16$ computational grid. Rows correspond to
    $\sigma_{11}$, $\sigma_{22}$, $\sigma_{12}$, $v_x$, and $v_y$.
    Columns show the exact, Lie-partitioned, classical Schr\"odingerised,
    and circuit solutions.}
    \label{fig:numerical_periodic_fields}
\end{figure}

We measure errors in the physical variables using
\begin{equation}
    \|e\|_h^2=h_xh_y\sum_{i,j}
    \left(
    |e_{\sigma_{11},ij}|^2+|e_{\sigma_{22},ij}|^2
    +|e_{\sigma_{12},ij}|^2+|e_{v_x,ij}|^2+|e_{v_y,ij}|^2
    \right).
    \label{eq:numerical_periodic_2d_norm}
\end{equation}
At $T=1$, the errors relative to the exact solution are listed in Tab.\ref{tab:error_convergence}.
The classical Schr\"odingerised result is more accurate because the matrix-exponential action is evaluated directly, avoiding temporal splitting errors. Both the classical and quantum-circuit methods exhibit first-order convergence. The remaining error differences are mainly due to different splitting commutators and the recovery errors in the circuit method. These results confirm the correctness of the proposed circuit.
\begin{table}[htbp]
\centering
\caption{Errors relative to the exact solution at $T=1$ on the $16\times16$ grid.}
\begin{tabular}{c c c c}
\hline
$N_t$
& Lie-partitioned method
& Classical Schr\"odingerised solution
& Circuit solution \\
\hline
$128$
& $3.5331\times10^{-2}$
& $6.8927\times10^{-6}$
& $2.1490\times10^{-2}$ \\

$256$
& $1.7664\times10^{-2}$
& $6.8927\times10^{-6}$
& $1.0750\times10^{-2}$ \\

$512$
& $8.8315\times10^{-3}$
& $6.8927\times10^{-6}$
& $5.3776\times10^{-3}$ \\
\hline
\end{tabular}
\label{tab:error_convergence}
\end{table}
\subsection{Homogeneous problem with impedance boundary conditions}

The second test uses the source-free system with unit impedances,
\begin{equation}
    -\sigma(-\pi,t)+v(-\pi,t)=0,
    \qquad
    \sigma(\pi,t)+v(\pi,t)=0.
    \label{eq:numerical_impedance_bc}
\end{equation}
Let
\begin{equation}
F(\xi)=
\begin{cases}
    (1-r^2)^4, & |r|<1,\\
    0, & |r|\geq 1,
\end{cases}
\qquad r=\frac{\xi+1}{0.8}.
\label{eq:numerical_impedance_pulse}
\end{equation}
The initial data $v(x,0)=F(x)$ and $\sigma(x,0)=-F(x)$ give the right-moving
exact solution
\begin{equation}
    v(x,t)=F(x-t),
    \qquad
    \sigma(x,t)=-F(x-t).
    \label{eq:numerical_impedance_exact}
\end{equation}
Over the simulated time interval, the support of this pulse remains separated
from both endpoints and therefore satisfies \eqref{eq:numerical_impedance_bc}.

The staggered discretization uses $N=63$ stress cells and $N+1=64$ velocity nodes.  Here $h=2\pi/N$, $Z_L=Z_R=1$,
$T=1.2$, and $N_t=512$.  The auxiliary parameters are $n_p=7$, $R=4$, and
$p_*=2.4$, giving a 14-qubit state-vector simulation.  

Fig.~\ref{fig:numerical_impedance} shows the recovered velocity and stress.
Using the staggered-grid weighted $L^2$ norm, the errors relative to the exact
solution are $3.5334\times10^{-2}$ for the Lie-partitioned method,
$3.5283\times10^{-2}$ for the classical Schr\"odingerised solution, and
$3.8984\times10^{-2}$ for the circuit solution.

\begin{figure}[H]
    \centering
    \includegraphics[width=0.98\linewidth]{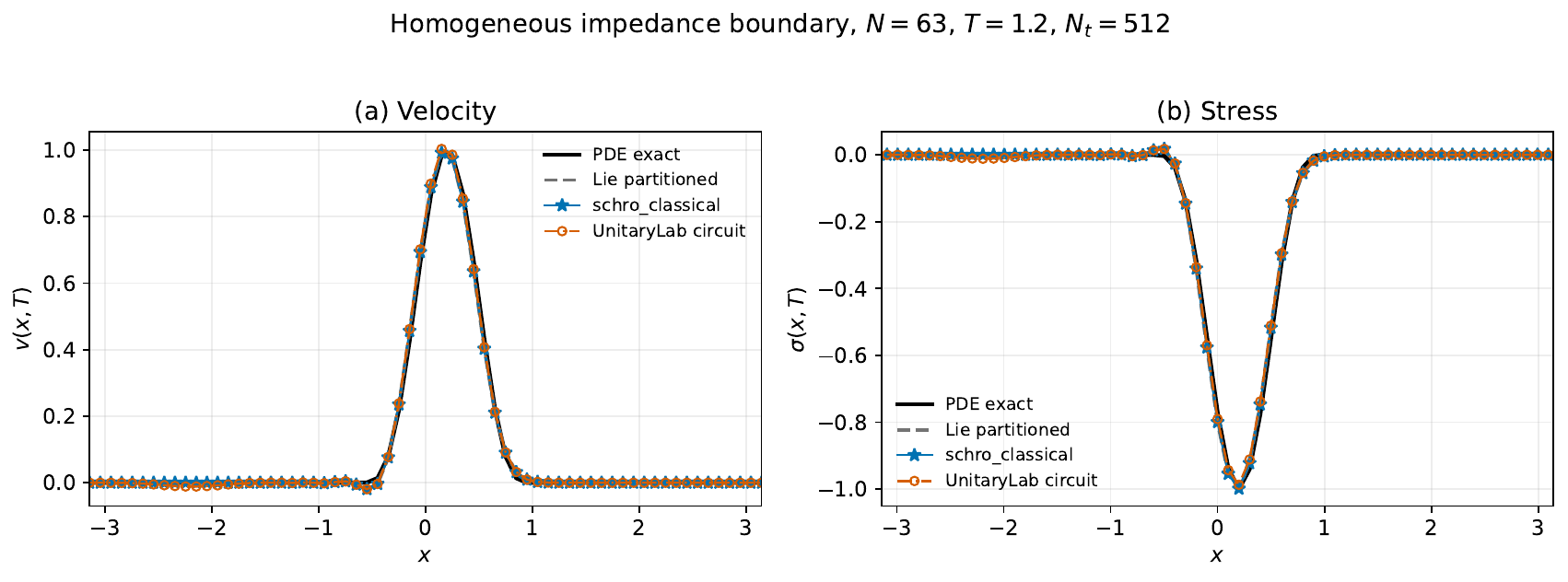}
    \caption{Velocity and stress for the homogeneous impedance problem at
    $T=1.2$.}
    \label{fig:numerical_impedance}
\end{figure}

\section{Conclusion}\label{sec:Conclusion}
We have developed explicit Schr\"odingerisation circuits for the
velocity--stress elastic wave equations in homogeneous media. The periodic
construction combines Fourier spectral discretization with controlled
rotations. For physical boundary conditions, staggered-grid discretization
and interleaved encoding lead to shift circuits and boundary corrections. These constructions accommodate time-independent
source terms and non-unitary boundary dynamics within the same Hamiltonian
simulation framework. 
Our results establish a theoretical foundation for quantum-accelerated modeling--with polynomial advantage over its classical counterparts for quantum state computation-- of elastic wave propagation, with potential applications in geophysical exploration and materials science.

Realistic wave propagation often involves complex boundary value problems and unbounded domains. Our future efforts will focus on quantum algorithms for such generalized scenarios, derive the corresponding quantum circuit for practical simulations and assess the costs of input preparation, observable
estimation, and implementation with a finite gate set.

\section*{Acknowledgement}

We acknowledge support from NSFC Grant No.~12341104, the Shanghai Pilot
Program for Basic Research, the STCSM Grant No.~24LZ1401200, the Shanghai
Jiao Tong University 2030 Initiative, and the Fundamental Research Funds
for the Central Universities. 
\bibliographystyle{elsarticle-num}
\bibliography{sample}

\onecolumn\newpage
\appendix

\section{Notations}
\subsection{Notation for Finite Difference Operators}\label{app:notations}
In a qubit representation, the finite difference operators can be written
in terms of matrix product operators. We first introduce the elementary
\(2\times2\) matrices: 
\begin{equation}
    \sigma_{01} := |0\rangle \langle 1|, \quad
    \sigma_{10} := |1\rangle \langle 0|, \quad
    \sigma_{00} := |0\rangle \langle 0|, \quad
    \sigma_{11} := |1\rangle \langle 1|,\quad I:=\sigma_{00}+\sigma_{11}.
\end{equation}
The left and right shift operators are defined by \(S_{n_x}^+|j\rangle =|j+1\rangle\), \(S_{n_x}^-|j\rangle = |j-1\rangle\),
\begin{equation}\label{eq:shift_operators}
\begin{aligned}
S_{n_x}^-
&=
\sum_{j=1}^{2^{n_x}-1}|j-1\rangle\langle j|
=
\sum_{\ell=1}^{n_x}
I^{\otimes(n_x-\ell)}
\otimes \sigma_{01}
\otimes \sigma_{10}^{\otimes(\ell-1)}
\triangleq
\sum_{\ell=1}^{n_x}s_\ell^-,
\\
S_{n_x}^+
&=
(S_{n_x}^-)^\dagger
=
\sum_{j=1}^{2^{n_x}-1}|j\rangle\langle j-1|
=
\sum_{\ell=1}^{n_x}
I^{\otimes(n_x-\ell)}
\otimes \sigma_{10}
\otimes \sigma_{01}^{\otimes(\ell-1)}
\triangleq
\sum_{\ell=1}^{n_x}s_\ell^+ .
\end{aligned}
\end{equation}
The rotation operators \(R_x(\varphi), R_y(\varphi), R_z(\varphi)\) are defined as: 
\begin{equation}\label{eq:def of R gate}
    R_x(\varphi)=e^{-i\varphi X/2},\quad
    R_y(\varphi)=e^{-i\varphi Y/2},\quad
    R_z(\varphi)=e^{-i\varphi Z/2}.
\end{equation}
\section{The related Proof}
\subsection{Proof of $\exp\left( \sum_{k} A_k \otimes |k\rangle \langle k| \right) = \sum_{k} \exp(A_k) \otimes |k\rangle \langle k|$}\label{app:pfA_k}
\begin{proof}
Let $A = \sum_k A_k\otimes \ket{k}\bra{k}$. 
Since the projectors $\{\ket{k}\bra{k}\}_k$ are mutually orthogonal,
\begin{equation}
    \ket{k}\bra{k}\ket{l}\bra{l}
    =
    \delta_{kl}\ket{k}\bra{k}.
\end{equation}
It follows that, for any positive integer $m$,
\begin{equation}
    A^m
    =
    \left(
    \sum_k A_k\otimes\ket{k}\bra{k}
    \right)^m
    =
    \sum_k A_k^m\otimes\ket{k}\bra{k}.
\end{equation}
Using the power-series expansion of the matrix exponential, we obtain
\begin{align}
    \exp(A)
    &=
    \sum_{m=0}^{\infty}\frac{1}{m!}A^m 
    =
    \sum_{m=0}^{\infty}\frac{1}{m!}
    \sum_k A_k^m\otimes\ket{k}\bra{k} \\
    &=
    \sum_k
    \left(
    \sum_{m=0}^{\infty}\frac{A_k^m}{m!}
    \right)
    \otimes\ket{k}\bra{k} 
    =
    \sum_k
    \exp(A_k)\otimes\ket{k}\bra{k}.
\end{align}
\end{proof}

\section{Constructions of \(\exp(iA_{\ell}t)\)}
We use the velocity and stress variables themselves to label the quantum basis states. That is, the Dirac states \(|v_i\rangle\), \(|n_i\rangle\), and \(|s_{ij}\rangle\) are identified with the corresponding components of the SMF system. In this basis, the coefficient matrix \(A_\ell\) can be expressed as
\[
    A_\ell
    =
    \sum_{j=1}^{d}
    a_{\ell j}
    \left(
        |n_j\rangle\langle v_\ell|
        +
        |v_\ell\rangle\langle n_j|
    \right)
    +
    \sum_{k\neq \ell}
    m_3
    \left(
        |s_{\ell k}\rangle\langle v_k|
        +
        |v_k\rangle\langle s_{\ell k}|
    \right),
\]
where \(s_{\ell k}=s_{k\ell}\) if \(k<\ell\), and
\[
    a_{\ell j}
    =
    \begin{cases}
    m_1, & j=\ell,\\
    m_2, & j\neq \ell.
    \end{cases}
\]
\subsection{The derivation of \(\exp(iA_{\ell}t)\)}\label{sec:derivation}
Since the directional operators $A_\ell, \ell=1,\ldots,d$
are generally noncommuting, the exponential of \(A\) cannot be factorized exactly. 
Thus, the \(d\)-dimensional circuit is obtained by sequentially implementing
the directional evolution operators
\[
\exp\left(i\tau A_\ell\otimes D_\ell\right),
\qquad
\ell=1,\ldots,d.
\]
Each factor acts only on the physical-component register and the spatial
register associated with the \(\ell\)-th direction, which makes the circuit
structure a direct extension of the one-dimensional construction.

From the velocity--stress formulation reviewed in Sec.~\ref{sec:Equ review}, it follows that \(A_\ell\) possesses a highly structured sparsity pattern. Specifically, \(A_\ell\) can be decomposed into three parts: a normal stress--velocity star block, \(d-1\) shear stress--velocity two-dimensional blocks, and a zero block.

Define \(\omega = \sqrt{m_1^2+(d-1)m_2^2},\)
and \(|\alpha_\ell\rangle =\frac{1}{\omega}
    \left(m_1 |n_\ell\rangle
        +
        m_2\sum_{j\neq \ell}|n_j\rangle
    \right).\)
Then the normal--velocity part of \(A_\ell\) can be written as
\begin{equation}
    B_\ell
    =
    \omega
    \left(
        |\alpha_\ell\rangle\langle v_\ell|
        +
        |v_\ell\rangle\langle \alpha_\ell|
    \right).
\end{equation}
Here $B_\ell$ is a rank-two Hermitian operator.
Let \(P_\ell\) be a unitary acting on the normal-stress subspace such that \(P_\ell |n_\ell\rangle = |\alpha_\ell\rangle,\)
and acting as the identity on the remaining variables. Then
\begin{equation}
    B_\ell
    =
    P_\ell
    \omega
    \left(
        |n_\ell\rangle\langle v_\ell|
        +
        |v_\ell\rangle\langle n_\ell|
    \right)
    P_\ell^\dagger.    
\end{equation}
Therefore, $e^{iB_\ell t}
    =
    P_\ell
    \mathcal X_{n_\ell,v_\ell}(\omega t)
    P_\ell^\dagger,$
where the two-level rotation \(\mathcal X_{p,q}(\theta)\) is defined by
\begin{equation}
    \mathcal X_{p,q}(\theta)
    =
    \exp\left(
        i\theta
        \left(
            |p\rangle\langle q|
            +
            |q\rangle\langle p|
        \right)
    \right).
\end{equation}
Restricted to the subspace \(\operatorname{span}\{|p\rangle,|q\rangle\}\),
it has the matrix representation
\begin{equation}
    \mathcal X_{p,q}(\theta)
    =
    \begin{pmatrix}
    \cos\theta & i\sin\theta\\
    i\sin\theta & \cos\theta
    \end{pmatrix}.
\end{equation}
The unitary \(P_\ell\) can be constructed using \(d-1\) Givens rotations in the
normal-stress subspace. Let
\[
    j_1=\ell,
    \qquad
    \{j_2,\ldots,j_d\}=\{1,\ldots,d\}\setminus\{\ell\}.
\]
Set \(r_k=\sqrt{m_1^2+(k-1)m_2^2}, k=1,\ldots,d.\)
Then \(r_d=\omega\). A sequence of Givens rotations with angles satisfying $\theta_k = \arctan \frac{m_2}{r_{k-1}}$, $k=2,\ldots,d$,
can be used to prepare the state \(|\alpha_\ell\rangle\) from
\(|n_\ell\rangle\). Thus \(P_\ell\) can be implemented by \(d-1\) two-level
real rotations, and \(P_\ell^\dagger\) by reversing the same rotations.

Next, for each \(k\neq \ell\), the shear--velocity part of \(A_\ell\) contains
the two-dimensional block
\[
    m_3
    \left(
        |s_{\ell k}\rangle\langle v_k|
        +
        |v_k\rangle\langle s_{\ell k}|
    \right).
\]
Its exponential is simply
\begin{equation}
    \mathcal X_{s_{\ell k},v_k}(m_3t)
    =
    \exp\left(
        i m_3 t
        \left(
            |s_{\ell k}\rangle\langle v_k|
            +
            |v_k\rangle\langle s_{\ell k}|
        \right)
    \right).
\end{equation}
Therefore, the full exponential of the directional matrix \(A_\ell\) admits
the exact factorization
\begin{equation}
    e^{iA_\ell t}
    =
    P_\ell
    \mathcal X_{n_\ell,v_\ell}(\omega t)
    P_\ell^\dagger
    \prod_{k\neq \ell}
    \mathcal X_{s_{\ell k},v_k}(m_3t).
\end{equation}
The decomposition is exact because the normal--velocity block and all
shear--velocity blocks act on mutually disjoint subspaces.

Therefore, the implementation of \(\exp(iA_\ell t)\) 
 consists of the basis transformation \(P_\ell\),  its inverse \(P_\ell^\dagger\), one two-level rotation associated with the normal--velocity block, and d-1 two-level rotations associated with the shear--velocity blocks.

\begin{remark}
The number of qubits used to encode the spatial variables and
physical components is
$m
=
dn+\left\lceil \log_2 q \right\rceil
=
dn+\left\lceil \log_2\frac{d^2+3d}{2} \right\rceil $.
For quantum implementation, we define the matrix \(\widetilde A_\ell
    =
    A_\ell\oplus0_{2^{\lceil\log_2q\rceil}-q}.\)
Then
\begin{equation}
    e^{i\widetilde A_\ell t}
    =
    e^{iA_\ell t}\oplus0_{2^{\lceil\log_2q\rceil}-q}.
\end{equation}
The extended variables can be set as an invariants.
\end{remark}
\subsection{The circuit of \(\exp(iA_{\ell}t)\)}\label{sec:circuit}
Directional circuits for $e^{i\widetilde A_\ell t}$ in  \(d=2,3\) are presented in Figs.~\ref{fig:Al-explicit-2d} and~\ref{fig:Al-explicit-3d}, respectively with
\begin{equation}
\begin{aligned}
    \theta_2&=\operatorname{atan2}(m_2,m_1),\qquad
\theta_3=\operatorname{atan2}\!\left(m_2,\sqrt{m_1^2+m_2^2}\right),\\
\omega_d=&\sqrt{m_1^2+(d-1)m_2^2},\qquad
\beta_P=-2\omega_dt,\qquad \beta_S=-2m_3t.
\end{aligned}
\end{equation}
Here $\theta_3$ is used only for $d=3$. All unused basis states
are left invariant. The circuits below are exact and no product-formula
approximation is used within a directional block.
\begingroup
\begin{figure}[h]
\centering
\(\begin{array}{c|ccccc}
\text{component}&n_{1}&n_{2}&s_{12}&v_{1}&v_{2}\\\hline
\text{basis label}&0&1&2&3&4
\end{array}\)
\par\smallskip
The basis label is the integer represented by $|q_{2}\cdots q_0\rangle$ for $d=2$.
\par\medskip
\textnormal{(a) $e^{i\widetilde A_1t}$}
\[
\Qcircuit @C=0.45em @R=0.85em {
\lstick{q_{2}} & \ctrlo{1} & \qw & \ctrlo{1} & \qw & \targ & \ctrl{1} & \targ & \ctrlo{1} & \qw \\
\lstick{q_{1}} & \ctrlo{1} & \targ & \ctrlo{1} & \targ & \ctrl{-1} & \gate{R_x(\beta_S)} & \ctrl{-1} & \ctrlo{1} & \qw \\
\lstick{q_{0}} & \gate{R_y(-2\theta_2)} & \ctrl{-1} & \gate{R_x(\beta_P)} & \ctrl{-1} & \qw & \ctrlo{-1} & \qw & \gate{R_y(2\theta_2)} & \qw \\
}
\]
\par\medskip
\textnormal{(b) $e^{i\widetilde A_2t}$}
\[
\Qcircuit @C=0.45em @R=0.85em {
\lstick{q_{2}} & \ctrlo{1} & \targ & \ctrl{1} & \targ & \ctrlo{1} & \ctrlo{1} & \qw \\
\lstick{q_{1}} & \ctrlo{1} & \qw & \ctrlo{1} & \qw & \ctrl{1} & \ctrlo{1} & \qw \\
\lstick{q_{0}} & \gate{R_y(2\theta_2)} & \ctrl{-2} & \gate{R_x(\beta_P)} & \ctrl{-2} & \gate{R_x(\beta_S)} & \gate{R_y(-2\theta_2)} & \qw \\
}
\]
\par\medskip
\caption{Directional circuits in $d=2$.
The $q_j$ wires encode the physical components, not the spatial variables.
}
\label{fig:Al-explicit-2d}
\end{figure}

\begin{figure}[h]
\centering
\(\begin{array}{c|ccccccccc}
\text{component}&n_{1}&n_{2}&n_{3}&s_{12}&s_{13}&s_{23}&v_{1}&v_{2}&v_{3}\\\hline
\text{basis label}&0&1&2&3&4&5&6&7&8
\end{array}\)
\par\smallskip
The basis label is the integer represented by $|q_{3}\cdots q_0\rangle$.
\par\medskip
\textnormal{(a) $e^{i\widetilde A_1t}$}
\[
\Qcircuit @C=0.45em @R=0.85em {
\lstick{q_{3}} & \ctrlo{1} & \ctrlo{1} & \qw & \ctrlo{1} & \qw & \ctrlo{1} & \targ & \ctrl{1} & \targ & \ctrlo{1} & \ctrlo{1} & \qw \\
\lstick{q_{2}} & \ctrlo{1} & \ctrlo{1} & \targ & \ctrlo{1} & \targ & \gate{R_x(\beta_S)} & \ctrl{-1} & \gate{R_x(\beta_S)} & \ctrl{-1} & \ctrlo{1} & \ctrlo{1} & \qw \\
\lstick{q_{1}} & \ctrlo{1} & \gate{R_y(-2\theta_3)} & \ctrl{-1} & \gate{R_x(\beta_P)} & \ctrl{-1} & \ctrl{-1} & \qw & \ctrlo{-1} & \qw & \gate{R_y(2\theta_3)} & \ctrlo{1} & \qw \\
\lstick{q_{0}} & \gate{R_y(-2\theta_2)} & \ctrlo{-1} & \qw & \ctrlo{-1} & \qw & \ctrl{-1} & \qw & \ctrlo{-1} & \qw & \ctrlo{-1} & \gate{R_y(2\theta_2)} & \qw \\
}
\]
\par\medskip
\textnormal{(b) $e^{i\widetilde A_2t}$}
\[
\resizebox{0.98\linewidth}{!}{$
\Qcircuit @C=0.45em @R=0.85em {
\lstick{q_{3}} & \ctrlo{1} & \qw & \ctrlo{1} & \qw & \qw & \ctrlo{1} & \qw & \qw & \ctrlo{1} & \qw & \qw & \targ & \ctrl{1} & \targ & \qw & \qw & \ctrlo{1} & \qw & \ctrlo{1} & \qw \\
\lstick{q_{2}} & \ctrlo{1} & \qw & \ctrlo{1} & \qw & \targ & \ctrlo{1} & \targ & \targ & \ctrl{1} & \targ & \targ & \qw & \ctrlo{1} & \qw & \targ & \qw & \ctrlo{1} & \qw & \ctrlo{1} & \qw \\
\lstick{q_{1}} & \ctrlo{1} & \targ & \ctrl{1} & \targ & \ctrl{-1} & \gate{R_x(\beta_P)} & \ctrl{-1} & \qw & \ctrl{1} & \qw & \qw & \qw & \ctrlo{1} & \qw & \qw & \targ & \ctrl{1} & \targ & \ctrlo{1} & \qw \\
\lstick{q_{0}} & \gate{R_y(2\theta_2)} & \ctrl{-1} & \gate{R_y(2\theta_3)} & \ctrl{-1} & \qw & \ctrl{-1} & \qw & \ctrl{-2} & \gate{R_x(\beta_S)} & \ctrl{-2} & \ctrl{-2} & \ctrl{-3} & \gate{R_x(\beta_S)} & \ctrl{-3} & \ctrl{-2} & \ctrl{-1} & \gate{R_y(-2\theta_3)} & \ctrl{-1} & \gate{R_y(-2\theta_2)} & \qw \\
}
$}
\]
\par\medskip
\textnormal{(c) $e^{i\widetilde A_3t}$}
\[
\Qcircuit @C=0.45em @R=0.85em {
\lstick{q_{3}} & \ctrlo{1} & \qw & \ctrlo{1} & \qw & \targ & \ctrl{1} & \targ & \ctrlo{1} & \ctrlo{1} & \qw & \ctrlo{1} & \qw & \ctrlo{1} & \qw \\
\lstick{q_{2}} & \ctrlo{1} & \qw & \ctrlo{1} & \qw & \qw & \ctrlo{1} & \qw & \ctrl{1} & \ctrl{1} & \qw & \ctrlo{1} & \qw & \ctrlo{1} & \qw \\
\lstick{q_{1}} & \gate{R_y(2\theta_2)} & \targ & \ctrl{1} & \targ & \ctrl{-2} & \gate{R_x(\beta_P)} & \ctrl{-2} & \gate{R_x(\beta_S)} & \gate{R_x(\beta_S)} & \targ & \ctrl{1} & \targ & \gate{R_y(-2\theta_2)} & \qw \\
\lstick{q_{0}} & \ctrlo{-1} & \ctrl{-1} & \gate{R_y(-2\theta_3)} & \ctrl{-1} & \qw & \ctrlo{-1} & \qw & \ctrlo{-1} & \ctrl{-1} & \ctrl{-1} & \gate{R_y(2\theta_3)} & \ctrl{-1} & \ctrlo{-1} & \qw \\
}
\]
\par\medskip
\caption{Directional circuits in $d=3$. 
The $q_j$ wires encode the physical components, not the spatial variables.
}
\label{fig:Al-explicit-3d}
\end{figure}
\endgroup

\subsection{The complexity of \(\exp(iA_{\ell}\tau)\)}\label{sec:complexity}
\begin{lemma}
    In the short time \(\tau\), the CNOT gate complexity of explicitly implementing      \(\exp(iA_\ell t)\) is bounded by $O(d(\log d))$.
\end{lemma}
\begin{proof}
By the decomposition derived in Sec.~\ref{sec:derivation}, the implementation of
\(\exp(iA_\ell t)\) involves the operators
\(P_\ell\), \(\mathcal X_{n_\ell,v_\ell}(\omega t)\), \(P_\ell^\dagger\), and
\(d-1\) additional two-level rotations \(\mathcal X_{s_{\ell k},v_k}(m_3t)\).

Consequently, each two-level rotation \(\mathcal X_{p,q}(\theta)\) can be implemented on the
computational basis by a multi-controlled \(R_x\) rotation. 
Similarly, each Givens rotation used in \(P_\ell\) can be implemented as a
two-level real rotation, corresponding to a multi-controlled \(R_y\) rotation
after mapping the two basis states to states differing in only one qubit by a
Gray-code transformation. Consequently, a total of \(3d-2\) two-level
\(R_x/R_y\) rotation gates are required.

According to Ref.~\cite{mottonen2004transformation}, each two-level
\(R_x/R_y\) rotation acting on an \(n\)-qubit register can be decomposed into
at most \(18n-42\) CNOT gates.
Therefore, the CNOT complexity for implementing
\(\exp(iA_\ell t)\) is bounded by
\[
Q_{\mathrm{CNOT}}\!\left(\exp(iA_\ell t)\right)
=
(3d-2)(18\left\lceil \log_2 q \right\rceil-42)
=
O(d\log d).
\]

\end{proof}

\end{document}